\documentclass[11pt,a4paper,oneside]{article}
\usepackage[top=3cm, bottom=3cm, left=2cm, right=2cm]{geometry}
\usepackage[english]{babel}
\usepackage[utf8]{inputenc}
\usepackage[T1]{fontenc}

\usepackage{amsthm,amsmath,amssymb,mathrsfs,dsfont,mathtools}
\usepackage{bbm}
\usepackage{multirow}
\usepackage{bm}
\usepackage{graphicx}
\usepackage[colorlinks=true, allcolors=blue]{hyperref}
\usepackage{comment}
\usepackage[all]{nowidow}
\usepackage{booktabs}
\usepackage{microtype}  
\usepackage{xcolor}

\usepackage{algorithm}
\usepackage{algpseudocode}

\usepackage[authoryear, round]{natbib}
\usepackage{authblk}
\date{}
\title{Nonparametric predictive inference for discrete data via Metropolis-adjusted Dirichlet sequences}
\author[1]{Davide Agnoletto}
\author[2]{Tommaso Rigon}
\author[1]{David B. Dunson}
\affil[1]{Department of Statistical Science, Duke University, Durham, NC, USA}
\affil[2]{Department of Economics, Management, and Statistics, University of Milano--Bicocca, 20126 Milano, Italy}

\providecommand{\keywords}[1]{
  \small 
  \textbf{\textit{Keywords:}} #1
  \normalsize
}

\usepackage[sf]{titlesec}
\usepackage{sectsty}
\allsectionsfont{\centering \normalfont\scshape} 

\newtheorem{theorem}{Theorem}
\newtheorem{corollary}{Corollary}
\newtheorem{lemma}{Lemma}
\newtheorem{proposition}{Proposition}
\theoremstyle{definition}

\newtheorem{definition}{Definition}
\newtheorem{remark}{Remark}
\theoremstyle{remark}

\usepackage{apptools}
\AtAppendix{\counterwithin{theorem}{section}}
\AtAppendix{\counterwithin{lemma}{section}}
\AtAppendix{\counterwithin{proposition}{section}}
\AtAppendix{\counterwithin{corollary}{section}}

\newcommand{\y}{\bm{y}}

\newcommand{\x}{\bm{x}}

\newcommand{\dd}{\mathrm{d}}
\newcommand{\var}{\mathrm{var}}
\newcommand{\E}{\mathds{E}}
\newcommand{\PP}{\mathds{P}}
\newcommand{\R}{\mathds{R}}
\newcommand{\N}{\mathds{N}}

\newcommand{\one}{\mathds{1}}

\begin{document}

\maketitle

\begin{abstract}
This article is motivated by challenges in conducting Bayesian inferences on unknown discrete distributions, with a particular focus on count data. To avoid the
computational disadvantages of traditional mixture models, we develop a novel Bayesian predictive approach. In particular, our Metropolis-adjusted Dirichlet (\textsc{mad}) sequence model characterizes the predictive measure as a mixture of a base measure and Metropolis-Hastings kernels centered on previous data points. The resulting \textsc{mad} sequence is asymptotically exchangeable and the posterior on the data generator takes the form of a martingale posterior. This structure leads to straightforward algorithms for inference on count distributions, with easy extensions to multivariate, regression, and binary data cases. We obtain a useful asymptotic Gaussian approximation and illustrate the methodology on a variety of applications.
\end{abstract}


\section{Introduction}
\label{sec:intro}
Bayesian nonparametric modeling of count distributions is a challenging task, and several solutions have been proposed over the years.
A Bayesian nonparametric mixture of Poisson kernels lacks flexibility since both centering and scaling of the kernels are controlled by a single parameter.
\cite{canale2011bayesian}, \citet{Canale2017} proposed a more flexible class of kernels obtained by rounding continuous distributions. However, estimating these  mixture models using Markov chain Monte Carlo (\textsc{mcmc}) algorithms can be cumbersome in practice.
Alternatively, one can directly specify a Dirichlet process (\textsc{dp}) prior \citep{ferguson1973} on the data generator as
\begin{equation}
Y_i\mid P \overset{\textup{iid}}\sim P, \qquad P\sim\text{DP}(\alpha, P_0),
\label{eqn:dp}
\end{equation}
where $Y_i\in\mathcal{Y}$, $i=1,\ldots,n$, are assumed to be independent and identically distributed (iid) conditionally on $P$, $\mathcal{Y}$ is a countable space, $\alpha$ is the precision parameter, and $P_0$ is a base parametric distribution, such as a Poisson.
Despite the appealing flexibility, it is well known that assuming \eqref{eqn:dp} produces the conjugate posterior distribution
\begin{equation*}
P\mid y_{1:n} \sim \mathrm{DP}\Big(\alpha+n, P_0 + \sum_{i=1}^n\delta_{y_i}\Big),
\end{equation*}
with $y_{1:n}=(y_1,\dots,y_n)$, which lacks smoothing and therefore tends to have poor performance unless the distribution is concentrated on a small number of values.
However, using traditional Bayesian nonparametric machinery, it is not clear how to specify a nonparametric process $P$ having the same simplicity and flexibility as \eqref{eqn:dp} while allowing smoothing. This motivates taking a predictive approach, which bypasses the need to directly specify a prior for $P$.

In a predictive approach, the statistical uncertainty is taken into account by modeling the conditional distribution of new data $Y_{n+1}$ given past observations $y_{1:n}$. \cite{blackwell1973ferguson} described the predictive rule resulting from the \textsc{dp} prior using a Pólya urn scheme for which $Y_1\sim P_0$ and $Y_{n+1}\mid y_{1:n}\sim P_n$ for $n\ge 1$, where
\begin{equation}
P_n(A) = \PP(Y_{n+1}\in A\mid y_{1:n})
= \frac{\alpha}{\alpha+n} P_0(A) + \frac{1}{\alpha+n}\sum_{i=1}^n\mathds{1}(y_i\in A),
\label{eqn:dp_predictive}
\end{equation}
for any measurable set $A \subseteq\mathcal{Y}$ and where $\one(\cdot)$ denotes the indicator function.
The sequence $(Y_n)_{n\ge1}$ generated according to \eqref{eqn:dp_predictive} is known as a Dirichlet sequence and is a conditionally iid~random sample from $P$, 
with ${P\sim\textsc{DP}(\alpha,P_0)}$. Although
predictive laws can be obtained by marginalizing models of the form in \eqref{eqn:dp}, potentially with alternative priors used in place of the \textsc{dp}, we instead start by defining the predictive law $P_n$. Thanks to Ionescu-Tulcea's theorem  one can characterize the law of the joint sequence $(Y_n)_{n \ge 1}$ by specifying a sequence of one-step-ahead predictive distributions. Moreover, if the resulting sequence is exchangeable, de Finetti's theorem implies the existence of a random probability measure $P$, conditionally on which the data are iid. This strategy bypasses prior elicitation; in our case, avoiding the difficult problem of directly specifying a prior for $P$.
Moreover, it highlights how the tasks of inference and prediction are intrinsically connected within the Bayesian paradigm.
There is a growing recent literature discussing the foundations of the Bayesian predictive approach, see for example \citet{fortini2012predictive, fortini2023prediction, fortini2024exchangeability}, \citet{holmes2023statistical}, 
\citet{fong_holmes_2023}, \citet{berti2025probabilistic}.

It is reasonable to think that a better estimator for count data distributions would be obtained by replacing $\mathds{1}(\cdot)$ in \eqref{eqn:dp_predictive} with a more diffuse kernel that allows the borrowing of information between the atoms. Early ideas along these lines can be found in \citet{hjort1994bayesian, hjort1994local}, whereas \citet{berti2023kernel} provide a rigorous recent discussion.
Unfortunately, for a general choice of the kernel, the corresponding sequence of predictive distributions does not generate an exchangeable sequence of random variables \citep{Sariev2024}, which is required for recovering the posterior distribution through de Finetti's theorem.
To address this problem, we consider a generalization of the predictive approach that relaxes the exchangeability assumption.

In this paper, we introduce a novel generalized Bayesian predictive method to model count data distributions.
We propose the recursive predictive rule
\begin{equation}
P_n(A) =(1-w_n) P_{n-1}(A) +  w_nK_n(A\mid y_n),
\label{eqn:kernel_dp}
\end{equation}
for any measurable set $A \subseteq \mathcal{Y}$, where $(w_n)_{n\ge1}$ is  a sequence of decreasing weights $w_n\in(0,1)$, and, for all $n\ge1$ we let $K_n(\;\cdot\mid y_n) = K(\;\cdot\;;y_n,P_{n-1})$
be a sequence of Metropolis-Hastings (\textsc{mh}) kernels, which will be formally defined in Section~\ref{sec:mad}. The special case \eqref{eqn:dp_predictive} is recovered for $w_n=(\alpha+n)^{-1}$ and $K_n(\cdot\mid y_n)=\one(y_n \in\cdot)$. The simple structure \eqref{eqn:kernel_dp} resembles a frequentist kernel density estimator (\textsc{kde}, \citealp{Wand1995}) under a particular choice of weights $w_n$, where \textsc{mh} kernels centered at the observed data points serve a similar role to local kernels in \textsc{kde}. More precisely, if $w_n = (\alpha +n)^{-1}$, then the recursion in equation~\eqref{eqn:kernel_dp} becomes
$$
P_n(A) = \frac{\alpha}{\alpha + n}P_0(A) + \frac{1}{\alpha + n}\sum_{i=1}^nK_i(A\mid y_i).
$$
This is similar in spirit to the recursive predictive algorithms of \citet{hahn2018recursive, fong_holmes_2023}, but, crucially, we specify a kernel $K_n(\cdot \mid y_n)$ tailored for count data. The predictive rule \eqref{eqn:kernel_dp} with \textsc{mh} kernels generates a Metropolis-adjusted Dirichlet (\textsc{mad}) sequence, which we show to be conditionally identically distributed (\textsc{cid}, \citealp{kallenberg1988spreading, berti2004limit}). In particular, \textsc{mad}s are \textsc{cid} sequences admitting a recursive structure \citep{berti2023without}. Consequently, \textsc{mad} sequences are asymptotically exchangeable, which means that, for large $n$, $(Y_n)_{n\ge1}$ can be regarded as an approximately iid sample from a random distribution $P$.
In other words, the predictive \eqref{eqn:kernel_dp} implies the existence of an implicit prior for $P$, which is a novel alternative to a \textsc{dp}. The posterior of $P$ can be regarded as a martingale posterior in the sense of \cite{fong_holmes_2023}, implying that a predictive resampling algorithm can be used to obtain posterior samples for $P$ allowing uncertainty quantification.

The literature about \textsc{cid} sequences \citep[e.g.][]{berti2004limit, Berti2012, Fortini2018, berti2021b, berti2023without, berti2025probabilistic} has strong roots in probability theory, but the statistical properties of such sequences have not been fully investigated, with a few notable exceptions. An approach based on the direct specification of copulas is discussed in \citet{hahn2018recursive, fong_holmes_2023}, whose focus is on continuous data with a Gaussian copula; see also \citet{cui2024martingale, fong2024bayesian, Bissiri2025} for related works. Such an update is not designed for count data, and when modeling multivariate distributions, it is dimension ordering dependent. As an alternative, \cite{fortini_petrone_2020} provided a Bayesian interpretation for the popular Newton recursive density estimator \citep{newton1999recursive, newton2002nonparametric} and proved that the resulting sequence is \textsc{cid}. Although, in principle, the Newton estimator can be adapted to the discrete case, the algorithm requires a numerical evaluation of an integral at each step, which could lead to computational bottlenecks if a closed form is not available;
using \eqref{eqn:kernel_dp} is substantially simpler.

We argue that employing \textsc{mh} kernels is particularly convenient for modeling count distributions for several reasons.
First, these kernels offer a clear interpretation as a mixture of an arbitrary kernel centered on past data and the \textsc{dp} update, which is recovered as a special case. 
This structure enables the borrowing of information between nearby values in the support when updating the predictive rule, resulting in a smoother estimator than the \textsc{dp}.
Second, computing the \textsc{mad} predictive is easy in practice, which is crucial since \eqref{eqn:kernel_dp} represents the posterior mean of $P$.
Consequently, obtaining a point estimate is efficient and does not require predictive resampling nor \textsc{mcmc}.
Third, the properties of the kernel facilitate extensions to the multivariate case and nonparametric regression, as the update remains invariant to the ordering of dimensions. This approach is also suitable for modeling multivariate binary data and proves particularly effective for nonparametric regression with binary covariates, as demonstrated through simulations and real data analyses.
Moreover, we obtain an asymptotic Gaussian approximation for the posterior distribution of $P$, highlighting the roles of the weights $w_n$ and the kernel bandwidth in the learning mechanism.

\section{Metropolis-adjusted Dirichlet sequences}\label{sec:mad}

We introduce a novel nonparametric predictive framework for modeling count data. Our goal is to develop a predictive construction that preserves the simplicity and flexibility of the Pólya urn scheme, while allowing for borrowing of information between atoms. Let $(Y_n)_{n\ge1}$ be a sequence of random variables defined on a countable space $\mathcal{Y}$. In the proposed approach, we assume that
\begin{equation}
Y_1\sim P_0, \qquad \text{and} \qquad Y_{n+1}\mid y_{1:n} \sim P_n,
\label{eqn:sequental_contruction}
\end{equation}
for $n\ge 1$, with $P_n(\cdot)=\PP(Y_{n+1}\in\cdot\mid y_{1:n})$ defined as in \eqref{eqn:kernel_dp} for some sequence of decreasing weights $(w_n)_{n \ge 1}$. In practice, we require $\sum_{n\ge1}w_n = \infty$ and $\sum_{n\ge1}w_n^2<\infty$ as this ensures good frequentist properties \citep{fortini_petrone_2020, fong_holmes_2023}, although this is not necessary to guarantee that $(Y_n)_{n \ge 1}$ is \textsc{cid}. The weights $w_n = (\alpha +n)^{-\lambda}$, with $\alpha>0$ and $\lambda\in(0.5,1]$, satisfy the required conditions. Importantly, this sequential construction leads to a unique and well-defined joint distribution $\mathds{P}$ for the entire sequence $(Y_n)_{n \ge 1}$ thanks to the celebrated Ionescu--Tulcea theorem \citep[e.g.][]{berti2025probabilistic}, which serves as the mathematical foundation of the predictive approach. To streamline the notation, we use capital letters $P_0$, $P_n(\cdot)$, and $K_n(\cdot \mid y_n)$ to denote probability measures, and lowercase $p_0(y)$, $p_n(y)$, and $k_n(y \mid y_n)$ for their corresponding probability mass functions. This convention will be used consistently throughout the paper. 

\begin{definition}
Let $\mathcal{Y}$ be a countable space and consider the recursive predictive rule~\eqref{eqn:kernel_dp} for $n \ge 1$
\begin{equation*}
p_n(y) = (1 - w_n)p_{n-1}(y) + w_n k_n(y\mid y_n), \qquad y \in \mathcal{Y},
\end{equation*}
with initial distribution $p_0(y)$ and weights $w_n \in (0,1)$. The Metropolis-Hastings kernel is defined as
\begin{equation}
    k_n(y\mid y_n) = \gamma_n(y,y_n) k_*(y\mid y_n) + \one(y=y_n)\Big[1 - \sum_{z\in\mathcal{Y}}\gamma_n(z,y_n)k_*(z\mid y_n)\Big],
    \label{eqn:mh_probability}
\end{equation}
where $k_*(y\mid y_n)$ is a discrete base kernel and  $\gamma_n(y, y_n)$ is a probability weight defined as
\begin{equation}
    \gamma_n(y,y_n) = \gamma(y, y_n, p_{n-1}) =
    \min\left\{1,\frac{p_{n-1}(y) k_*(y_n\mid y)}{p_{n-1}(y_n) k_*(y\mid y_n)}\right\}.
    \label{eqn:acc_probability}
\end{equation}
The predictive distribution $p_n$ is referred to as the Metropolis-adjusted Dirichlet (\textsc{mad}) distribution with weights $(w_n)_{n \ge 1}$, base kernel $k_*$, and initial distribution $p_0$. The corresponding sequence of random variables $(Y_n)_{n \ge 1}$, as defined in~\eqref{eqn:sequental_contruction}, is called a \textsc{mad} sequence.

\label{def:mad}
\end{definition}

The kernel probability mass function \eqref{eqn:mh_probability} can be interpreted as a mixture of a re-weighted discrete base kernel and a point mass at $y_n$. Specifically, for each $y\in\mathcal{Y}$, the base kernel probability $k_*(y\mid y_n)$ is weighted by $\gamma_n(y,y_n)$, which reflects how well it aligns with $p_{n-1}(y)$, as defined in \eqref{eqn:acc_probability}.
Since ${\sum_{y\in\mathcal{Y}}\gamma_n(y,y_n)k_*(y\mid y_n)\le 1}$, the remaining probability mass is assigned to $y_n$.
When $\gamma_n(y,y_n)=0$ for every $y\in\mathcal{Y}$, we obtain $k_n(y\mid y_n)=\one(y=y_n)$, and the \textsc{dp} predictive distribution is recovered if $w_n=(\alpha+n)^{-1}$.
At the other extreme, when $\gamma_n(y,y_n)=1$ for every $y\in\mathcal{Y}$, we obtain $k_n(y\mid y_n) = k_*(y\mid y_n)$.
For intermediate values of $\gamma_n(y,y_n)$, the \textsc{mh} kernel adjusts the Polya-urn scheme, introducing a smooth deviation from the \textsc{dp} update. This flexibility makes \textsc{mad} sequences particularly appealing for modeling discrete distributions. 

\subsection{Bayesian properties of MAD sequences}
\label{sec:bayes}

We provide a formal Bayesian justification for using a \textsc{mad} sequence to model discrete data. Although \textsc{mad} sequences are not exchangeable, they are conditionally identically distributed (\textsc{cid}). The notion of \textsc{cid} sequences was introduced by \citet{kallenberg1988spreading} as a relaxation of exchangeability, and it has been further developed by \citet{berti2004limit}. A sequence $(Y_n)_{n \ge 1}$ of random variables is said to be \textsc{cid} if $\mathds{P}(Y_n \in \cdot) = P_0(\cdot)$ for every $n \ge 1$, meaning that each $Y_n$ has the same marginal distribution $P_0$ a priori, and
\begin{equation}
    \PP(Y_{n+k} \in \cdot \mid y_{1:n}) = \PP(Y_{n+1} \in \cdot \mid y_{1:n}) = P_n(\cdot),
    \qquad \text{for all } k \ge 1, \ n \ge 1.
    \label{eqn:cid}
\end{equation}
It is sufficient to verify this condition for $k = 1$ in order to ensure its validity for all $k \ge 1$. Equivalently, \citet{fong_holmes_2023} express condition~\eqref{eqn:cid} in a way that emphasizes the martingale property of the predictive distributions $(P_n)_{n \ge 1}$:
\begin{equation}
    \mathds{E}\{P_{n+1}(\cdot) \mid y_{1:n}\} = P_n(\cdot), \qquad n \ge 1.
    \label{eqn:martingale}
\end{equation}
Note that $\mathds{E}\{P_{n+1}(\cdot) \mid y_{1:n}\} = \mathds{P}(Y_{n+2} \in \cdot \mid y_{1:n})$, which highlights the equivalence between the martingale and \textsc{cid} conditions. See \citet[][Theorem 3.1]{Berti2012} for a formal proof.

\begin{theorem}
\label{th:cid}
Let $(Y_n)_{n \ge 1}$ be a \textsc{mad} sequence. Then, for every choice of weight sequence $(w_n)_{n \ge 1}$, discrete base kernel $k_*$, and initial distribution $p_0$, the sequence $(Y_n)_{n \ge 1}$ is \textsc{cid}.
\end{theorem}

Theorem~\ref{th:cid} follows as a special case of \citet[][Theorem 2]{berti2023without}, since \textsc{mad} sequences fall within the broad class of predictive updates based on stationary kernels.  Let $\theta = P(f) = \sum_{y \in \mathcal{Y}} f(y) p(y)$ denote any functional of interest, and analogously let $\theta_n = P_n(f) = \sum_{y \in \mathcal{Y}} f(y) p_n(y)$ denote its predictive counterpart.

\begin{corollary}[\citealp{berti2004limit}]
\label{cor:cid}
Consider a \textsc{mad} sequence $(Y_n)_{n\ge1}$. Then, $\PP$-a.s.,
\begin{itemize}
\item[$(a)$] The sequence is asymptotically exchangeable, that is 
\begin{equation*}
    (Y_{n+1}, Y_{n+2}, \ldots) \overset{\textup{d}}{\longrightarrow} (Z_1, Z_2, \ldots), \qquad n \rightarrow \infty,
\end{equation*}
where $(Z_1,Z_2,\ldots)$ is an exchangeable sequence with directing random probability measure $P$;
\item[$(b)$] The corresponding sequence predictive distributions $(P_n)_{n\ge1}$ and empirical distributions $(\Hat{P}_n)_{n\ge1}$, where $\Hat{P}_n := n^{-1}\sum_{i=1}^n \delta_{Y_i}$, converge weakly to the random probability measure $P$;
\item[$(c)$] For every $n\ge 1$ and every integrable function $f:\mathcal{Y}\rightarrow\R$, we have $\mathds{E}(\theta) = \theta_0$ and $\mathds{E}(\theta \mid y_{1:n}) = \theta_n$.
\end{itemize}
\end{corollary}

These results ensure that \textsc{mad} sequences are a valid Bayesian method for nonparametric inference, bypassing direct prior specification on the space of probability distributions.  The results of Corollary~\ref{cor:cid} are based on the general results provided by \cite{berti2004limit} for \textsc{cid} sequences.  Although the probability law implied by a \textsc{mad} sequence is not ordering-invariant, part~(a) of Corollary~\ref{cor:cid} guarantees that this dependence vanishes asymptotically. As a result, for large $n$, exchangeability is recovered and, informally, we will say $Y_i\mid P \overset{\textup{iid}}{\sim} P$ for large $n$. Thus, an asymptotic equivalent of de Finetti’s theorem holds: each \textsc{mad} sequence has a corresponding unique prior on $P$ and part $(b)$ guarantees that $P$ is defined as the limit of the predictive distributions. This convergence holds in the weak sense, but it can be strengthened to total variation convergence (see Lemma~\ref{lemma_1} in Appendix~\ref{appendix_proofs}). Finally, by part $(c)$ we have $\theta = \E\{P(f)\}=P_0(f) = \theta_0$ for every integrable function $f$, so that $P_0$ retains the role of a base measure as for standard Dirichlet sequences, providing an initial guess at $P$. In particular, if $f(y) = \mathds{1}_A(y)$ is the indicator function for some measurable $A \subseteq \mathcal{Y}$, then we get $\mathds{E}\{P(A)\} = P_0(A)$. Consequently, prior or external information can still be incorporated into the model through $P_0$. Theorem~\ref{th:cid} and Corollary~\ref{cor:cid} are valid for a general sequence of weights $(w_n)_{n\ge1}$. When $w_n=(\alpha+n)^{-\lambda}$, the parameter $\alpha$ retains the same interpretation as the precision parameter in the \textsc{dp}, controlling the degree of shrinkage of $P$ toward $P_0$. Setting large values for $\alpha$ denotes a strong confidence in the base measure, while $\alpha=1$ provides a useful weakly informative default value. For a discussion of the role of $\lambda$ based on asymptotic considerations, we refer to Section~\ref{sec:asymptotics}. 

Given the data $y_{1:n}$, the predictive distributions $P_n$ converge to the posterior distribution for $P$, which can be considered a martingale posterior in the sense of \citet{fong_holmes_2023}. Prior and posterior sampling of $P$ can be performed using a simple Monte Carlo scheme (see Section~\ref{sec:martingale}). To illustrate, if the sequence of predictive distributions is 
equation \eqref{eqn:dp_predictive}, then the output of the posterior sampling algorithm for a given set $A \subseteq \mathcal{Y}$ would be an approximate iid sample from $P(A)\mid y_{1:n}$, whose distribution is
\begin{equation*}
\mathrm{Beta}\Big(
\alpha P_0(A)+\sum_{i=1}^n \mathds{1}(y_i\in A),\;
\alpha + n - \alpha P_0(A)-\sum_{i=1}^n \mathds{1}(y_i\in A)
\Big),
\end{equation*}
that is, the distribution of $P(A)\mid y_{1:n}$ when assuming a \textsc{dp} prior.  However, to obtain a point estimate for any functional of interest $\theta = P(f)$, sampling is not necessary. By part~(c), the posterior mean is $\mathds{E}(\theta \mid y_{1:n}) = \mathds{E}\{P(f) \mid y_{1:n}\} = P_n(f) = \theta_n$, which is a consequence of the martingale property. 

Although the posterior variance of $P(f)$ depends on the unknown limit distribution $P$, examining the variance of $P_{n+1}(f)$ given the observed data provides useful insights into \textsc{mad} sequences. The expected variation in the predictive distribution after observing a new data point is expressed as
\begin{equation}
    \var\{P_{n+1}(f)\mid y_{1:n}\} = \E\left\{[P_{n+1}(f)-P_n(f)]^2\mid y_{1:n}\right\}
    = w_{n+1}^2 \bigg\{\sum_{y\in\mathcal{Y}}K_n(f\mid y)p_n(y)-P_n(f)^2\bigg\}.
    \label{eqn:var_1step}
\end{equation}
When no data are observed, \eqref{eqn:var_1step} represents the expected prior-to-posterior variability after the first observation. We will show in Section~\ref{sec:asymptotics} that the variance of $P_{n+1}$ is also closely related to the asymptotic variance of $P$. Moreover, the correlation between two measurable functions $f_1, f_2$ induced by $P_{n+1}$ given $y_{1:n}$ is 
 \begin{equation}
 \begin{aligned}
    &\mathrm{cor}\{P_{n+1}(f_1),P_{n+1}(f_2)\mid y_{1:n}\} 
    = \\
    &\qquad=\frac{\sum_{y\in\mathcal{Y}}K_n(f_1\mid y)K_n(f_2\mid y)p_n(y) - P_n(f_1)P_n(f_2)}{\sqrt{\big[\sum_{y\in \mathcal{Y}}K_n(f_1\mid y)p_n(y)-P_n(f_1)^2\big]\big[\sum_{y \in \mathcal{Y}}K_n(f_2\mid y)p_n(y)-P_n(f_2)^2\big]}}.
    \label{eqn:cor_1step}
\end{aligned}
\end{equation}

\begin{remark}
In general, the predictive law $P_n$ of a \textsc{cid} sequence depends on the order of the observations~$y_{1:n}$, which is typically unappealing. In contrast, the predictive law $P_n$ of an exchangeable sequence does not depend on the order of $y_{1:n}$~\citep{fortini2000exchangeability}. This raises the question of whether a \textsc{cid} sequence whose predictive distribution is order-invariant even exists. A profound and elegant answer is provided in \citet[Corollary 3.7]{Fortini2018}, which establishes that a \textsc{cid} sequence is exchangeable if and only if its predictive law does not depend on the ordering of the observations. In other words, order dependence in martingale posteriors is unavoidable---otherwise, we would revert to the classical exchangeable framework, which is often too restrictive.

\end{remark}

\subsection{On the choice of the base kernel} 
\label{sec:rounded_gaussian}

In this section, we consider the univariate space $\mathcal{Y} = \{0,1,\dots\}$, whereas in Section~\ref{sec:multiv_counts} we will address the multivariate case. Although the results in Section~\ref{sec:bayes} hold for an arbitrary discrete base kernel $k_*$, the choice of kernel has important consequences for the \textsc{mh} update. For example, using Poisson base kernels—where a single parameter governs both the location and the scale—limits the flexibility of the \textsc{mh} kernel, as the same parameter simultaneously controls both the bandwidth and the center of the distribution. A simple alternative, resembling common choices in \textsc{kde}, is to consider a uniform kernel over the interval $\{y_n - m, \dots, y_n + m\}$ for some $m \ge 0$, that is, $k_*(y \mid y_n, m) =  (2m+1)^{-1}$ if $y \in \{y_n - m, \dots, y_n + m\}$ and $0$ otherwise. Then, the resulting \textsc{mh} kernel is centered at $y_n$, with bandwidth controlled by $m$. The corresponding probability mass function becomes
\begin{equation*}
k_n(y \mid y_n, m) = \min\left\{1, \frac{p_n(y)}{p_n(y_n)}\right\} k_*(y \mid y_n, m) + \one(y = y_n)\left[1 - \frac{1}{2m+1} \sum_{z = y_n - m}^{y_n + m} \min\left\{1, \frac{p_n(z)}{p_n(y_n)}\right\}\right],
\end{equation*}
for every $n \ge 1$, which has finite support $\{y_n - m, \dots, y_n + m\}$. When $m = 0$, the \textsc{dp} predictive distribution is recovered. Another simple possibility, paralleling classical choices in \textsc{mcmc}, is to choose a base kernel $k_*(y\mid y_n) = k_*(y)$ that does not depend on past observations, with $k_*(y)$ being a known distribution such as a Poisson or a negative binomial. This specification enforces a strong parametric structure and is suitable only when small deviations from a known model are expected. Although these two proposals lead to substantial computational simplifications, using more flexible kernels can improve predictive accuracy.

\begin{figure}[t]
\centering
\includegraphics[width=0.9\textwidth]{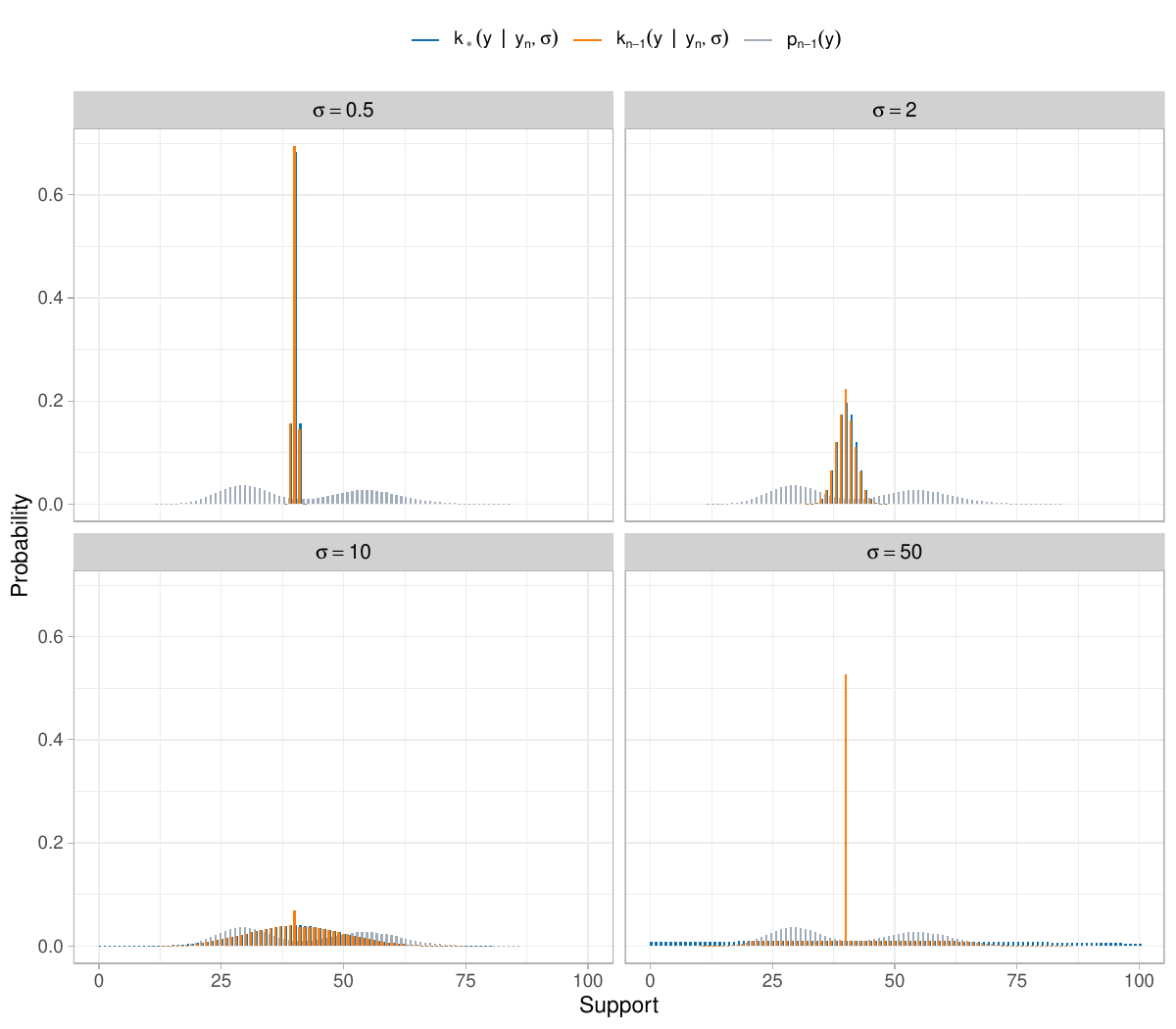}
\caption{\small Example of the probability mass function of a \textsc{mh} kernel using a rounded Gaussian base kernel with $\sigma\in\{0.5, 2, 10, 50\}$ after observing $y_{n}=40$.}
\label{fig:rg_kernel}
\end{figure}

\begin{figure}[t]
\centering
\includegraphics[width=0.95\textwidth]{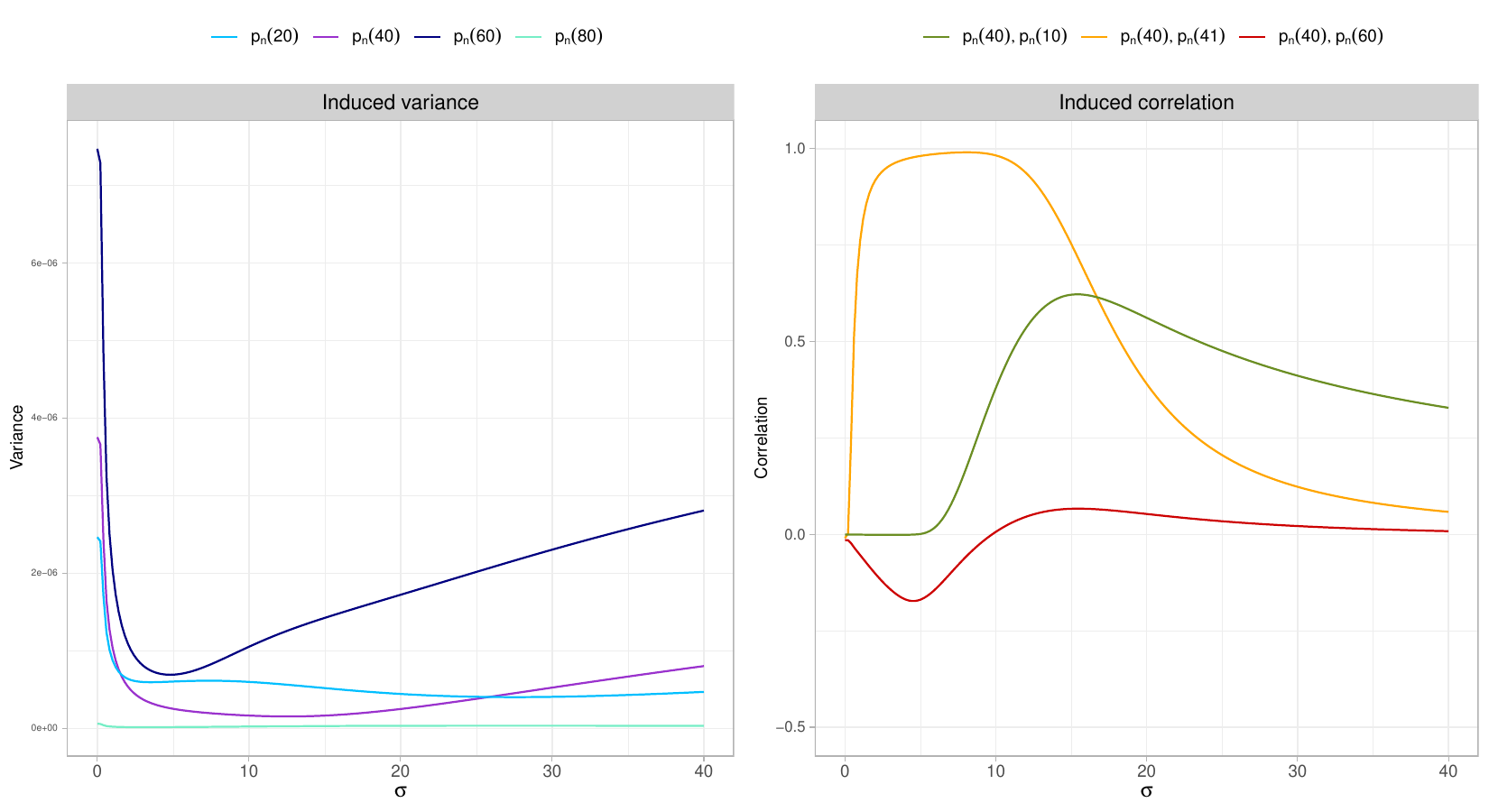}
\caption{\small Left: posterior variance of $p_{n}(z)$ given $y_{1:{n-1}}$, with $z\in\{20,40,60,80\}$.
Right: correlation between $p_{n}(40)$ and $p_{n}(x)$ given $y_{1:n-1}$, with $x\in\{10,41,60\}$.
For both panels, $P_{n-1}$ is the same as in Figure~\ref{fig:rg_kernel}.}
\label{fig:var_cor}
\end{figure}

As an alternative, we recommend using a rounded Gaussian base kernel as in \citet{canale2011bayesian}, whose probability mass function is given by
\begin{equation*}
    k_*(y\mid y_n, \sigma) = \frac{\int_{y-1/2}^{y+1/2}\mathcal{N}(t\mid y_n, \sigma^2) \mathrm{d}t}{\sum_{z\in\mathcal{Y}} \int_{z-1/2}^{z+1/2}\mathcal{N}(t\mid y_n, \sigma^2) \mathrm{d}t},
\end{equation*}
for $n\ge1$, where $\mathcal{N}(\cdot\mid y_n,\sigma^2)$ denotes a normal density function with mean $y_n$ and variance $\sigma^2$. Although $\sigma$ controls the bandwidth of the rounded Gaussian base kernel $k_*$, this role changes slightly after the Metropolis adjustment, as illustrated in Figure~\ref{fig:rg_kernel}. When $\sigma\rightarrow0$, the \textsc{mh} kernel collapses to a point mass at $y_n$, recovering the \textsc{dp} update $k_n(y\mid y_n,\sigma)=\one(y=y_n)$. A similar behavior occurs as $\sigma \to \infty$, when the base kernel becomes very diffuse: for values of $y$ such that $p_{n-1}(y)$ is close to zero $\gamma_n(y,y_n)\rightarrow 0$ and consequently the
probability mass on $y_n$, $1 - \sum_{z\in\mathcal{Y}}\gamma_n(z,y_n)k_*(z\mid y_n,\sigma)$, increases, recovering again the \textsc{dp} update.
For intermediate values of $\sigma$, the \textsc{mh} kernel is smoother, with $k_n(y\mid y_n,\sigma) = k_*(y\mid y_n,\sigma)$ when $\gamma_n(y,y_n) = 1$ for all $y\in\mathcal{Y}$. The bandwidth $\sigma$ also impacts the posterior variability. The left panel of Figure~\ref{fig:var_cor} illustrates how $\sigma$ influences the variance of $P_{n+1}\mid y_{1:n}$.
For $\sigma=0$, the \textsc{dp} is recovered, and the variance is maximum.
In this case, \eqref{eqn:var_1step} is proportional to the posterior variance induced by the \textsc{dp}.
As $\sigma$ grows, the kernel becomes more diffuse and the variance dramatically decreases. The variance increases again for large values of $\sigma$, corresponding to concentrated \textsc{mh} kernels as discussed above. The bandwidth also has important consequences on the correlation. Unlike in the \textsc{dp} case, the induced correlation between any pair of disjoint sets $A_1, A_2$ for \textsc{mad} sequences is not restricted to be negative, as illustrated in the right panel of Figure~\ref{fig:var_cor}. This is a very appealing feature, because positive correlation enables borrowing of information between nearby values.

\subsection{Uncertainty quantification via predictive resampling}
\label{sec:martingale}

A predictive resampling algorithm \citep{fortini_petrone_2020,fong_holmes_2023}, described in Algorithm~\ref{alg:predictive_resampling}, can be used for approximate posterior sampling of $P$, whose law, unfortunately, is generally not available in closed form. Algorithm~\ref{alg:predictive_resampling} can also be used for prior sampling of $P$ by setting $n = 0$. Starting from $P_n$, the algorithm performs forward simulations by recursively sampling and updating the \textsc{mad} predictive distribution. 
Each forward simulation generates a random dataset $(Y_{n+1},\dots,Y_{N})$ which can be used to compute $P_N(\cdot)$ or any other statistic of interest, such as the population mean. By part $(b)$ of Corollary~\ref{cor:cid}, we can regard $P_N$ as a finite-sample approximation of the limit distribution $P$ for sufficiently large $N$. Consequently, the output of Algorithm~\ref{alg:predictive_resampling} provides an approximate sample from the distribution of
$P(\cdot)\mid y_{1:n}$. 

Clearly, it is also possible to perform posterior inference on any functional $\theta = P(f)$.
For example, if one is interested in the population mean $\theta = P(f) = \sum_{y\in\mathcal{Y}} y \, p(y)$, then
the corresponding martingale posterior distribution can be approximated by considering its finite-sample version $\theta_N = P_N(f) = \sum_{y\in\mathcal{Y}} y\, p_N(y)$ which can be easily simulated using Algorithm~\ref{alg:predictive_resampling}. However, as previously mentioned, the posterior mean does not require approximations nor sampling, as we have $\mathbb{E}(\theta \mid y_{1:n}) = \theta_n$ thanks to the martingale property.

\begin{remark} The discreteness of the \textsc{mad} predictive rule is computationally advantageous compared to the continuous case \citep[e.g.][]{fong_holmes_2023}. First, despite what the Metropolis--Hastings terminology might suggest, the recursive update of $p_n$ does not involve any sampling, as $p_n$ can be directly updated using the definition of the kernel $k_n(y \mid y_n)$. This update is straightforward when both $p_0$ and the base kernel $k_*$ have finite support, since $p_n$ can then be represented by a finite set of probability weights. Sampling from $p_n$ is also straightforward, again because it involves a finite set of probabilities. In the countable case, it suffices to truncate $p_n$ by removing support points with negligible probability, which is effectively equivalent to using a base kernel with finite support. 
\end{remark}

\begin{algorithm}[t]
\caption{Predictive resampling \citep{fortini_petrone_2020, fong_holmes_2023}}
\label{alg:predictive_resampling}
\begin{algorithmic}
\State Compute $p_n$ as in \eqref{eqn:kernel_dp} from the observed data $y_{1:n}$
\State Set $N\gg n$
\For{$j = 1,\ldots,B$}
	\For{$i=n+1,\ldots,N$}
		\State Sample $Y_i\mid y_{1:i-1}\sim p_{i-1}$
		\State Update $p_i(y) = (1-w_i)p_{i-1}(y) + w_i k_i(y \mid y_i)$
	\EndFor
\EndFor
\State \Return $p_N^{(1)}(y),\ldots,p_N^{(B)}(y)$, an iid sample of size $B$ from the distribution of $p_N(\cdot)\mid y_{1:n}$
\end{algorithmic}
\end{algorithm}

\subsection{Asymptotic behaviour}
\label{sec:asymptotics}

In this section, we establish a predictive central limit theorem for \textsc{mad} sequences. Specifically, for any $H\ge1$ and any measurable sets $A_1,\ldots,A_H$, we derive a Gaussian approximation for the posterior distribution of $P(A_1),\ldots,P(A_H)\mid y_{1:n}$ as $n\rightarrow\infty$.
Intuitively, our result quantifies the uncertainty that arises from observing a finite sample of size $n$ instead of the infinite population. Within the Bayesian predictive literature, \cite{fortini_petrone_2020} provide a predictive central limit theorem for the Newton estimator, while \cite{fong2024asymptotics} obtain a similar result for parametric martingale posteriors. The following results specialize those of \cite{fortini2024exchangeability} for general \textsc{cid} sequences; however, in the context of \textsc{mad} sequences, some of their assumptions are not required, and therefore the proof is provided. Consistent with \cite{fortini_petrone_2020, fortini2024exchangeability} and \cite{fong2024asymptotics}, the proof relies on a martingale central limit theorem formulated in terms of almost sure conditional convergence \citep{crimaldi2009almost}, a notion of convergence that implies stable convergence \citep{renyi1963stable, hausler2015stable}.
Stable convergence, which is stronger than convergence in distribution but weaker than convergence in probability, plays a central role in martingale central limit theorems, especially when the limiting variance is random.

\begin{proposition}
\label{prop:asy_1}
Let $(Y_n)_{n\ge1}$ be a \textsc{mad} sequence and let $(r_n)_{n\ge1}$ be a monotone sequence of positive numbers such that $r_n\overset{\cdot}{=}(\sum_{k>n}w_k^2)^{-1}$ as $n\rightarrow\infty$.
Assume $\sqrt{r_n}\sup_{k\ge n}w_k\rightarrow 0$, $\sum_{k\ge1}r_k^2w_{k+1}^4<\infty$, and that, for every $H\ge1$ and all measurable sets $A_1,\ldots,A_H$, the random matrix $\bm\Sigma = \bm\Sigma(A_1,\ldots,A_H) = [\Sigma_{jt}]$, $j,t=1,\ldots,H$, with 
\begin{equation}
    \Sigma_{jt} = \sum_{y\in\mathcal{Y}}K(A_j\mid y)K(A_t\mid y)p(y) - P(A_j)P(A_t),
    \label{eqn:asy_variance}
\end{equation}
where $K(A\mid y) = \sum_{x\in A}k(x\mid y)$ and
\begin{equation*}
    k(x\mid y) = \gamma(x,y)k_*(x\mid y) + \one(x = y)\left[1-\sum_{z\in\mathcal{Y}}\gamma(z,y)k_*(z,y)\right],\quad
    \gamma(x,y) = \min\left\{1,\frac{p(x)k_*(y\mid x)}{p(y)k_*(x\mid y)}\right\},
\end{equation*}
is positive definite.
Then, for every $H\ge1$ and every $A_1,\ldots,A_H$, 
\begin{equation*}
    \sqrt{r_n}\left[\begin{array}{c}
     P(A_1)-P_n(A_1)\\
     \cdots\\
     P(A_H)-P_n(A_H)
    \end{array}\right] \Big|\; Y_{1:n}
    \overset{\textup{d}}{\longrightarrow} \bm{Z}, \qquad \bm{Z} \sim \mathcal{N}_H(\bm0,\bm\Sigma),
\end{equation*}
$\PP$-a.s. for $n\rightarrow\infty$,
where $\mathcal{N}_H(\bm0,\bm\Sigma)$ denotes a $H$-dimensional Gaussian distribution with mean $\bm0$ and covariance $\bm\Sigma$.
\end{proposition}

The proposition shows that the posterior distribution of $P$ concentrates around its expectation $P_n$ at a convergence rate of
$1/\sqrt{r_n}$. The asymptotic variance $\bm\Sigma$ is random, as it depends on the entire sequence $(Y_1,Y_2,\ldots)$. Having a random limit for the covariance is common in martingale central limit theorems and Bayesian predictive inference (see \citealp{fortini_petrone_2020, fortini2024exchangeability, berti2021b, fong2024asymptotics}).
This result should not be regarded as a Bernstein–von Mises asymptotic Gaussian approximation for the posterior distribution, because Proposition~\ref{prop:asy_1} is formulated under the joint law $\PP$ and does not assume the $Y_i$ are iid from a true underlying distribution.

Proposition~\ref{prop:asy_1} provides insight into the prior variance induced by \textsc{mad} sequences: if $n=0$ then $\bm\Sigma$ corresponds to the implicit prior variance, which depends on the limit distribution $P$. 
For $n\ge1$, the interpretation is less straightforward, as the asymptotic variance now depends on both past and future observations.
However, each element of the random covariance in \eqref{eqn:asy_variance} is proportional to the $\PP$-a.s. limit of the previously considered finite-sample covariance
\begin{equation} 
\mathrm{cov}\{P_{n+1}(A_j),P_{n+1}(A_t)\mid y_{1:n}\} = w_{n+1}^{2} \sum_{y\in\mathcal{Y}}K_n(A_j\mid y)K_n(A_t\mid y)p_n(y) - P_n(A_j)P_n(A_t), 
\label{eqn:cov_1step} 
\end{equation}
for any $j,t=1,\ldots,H$.
Consequently, we can employ \eqref{eqn:cov_1step} to approximate $\bm\Sigma$.

\begin{proposition}
\label{prop:asy_2}
For every $n\ge1$ and all measurable sets $A_1,\ldots,A_H$, let $\bm\Sigma_n=\bm\Sigma_n(A_1,\ldots,A_H) = [\Sigma_{n,jt}]$, $j,t=1,\ldots,H$, with
\begin{equation*}
\Sigma_{n,jt} = \sum_{y\in\mathcal{Y}}K_n(A_j\mid y)K_n(A_t\mid y)p_n(y) - P_n(A_j)P_n(A_t).
\end{equation*}
Then, under the assumptions of Proposition~\ref{prop:asy_1}, $\PP$-a.s., $\bm\Sigma_n\rightarrow\bm\Sigma$ and
\begin{equation*}
    \sqrt{r_n}\bm\Sigma_n^{-1/2}\left[\begin{array}{c}
     P(A_1)-P_n(A_1)\\
     \cdots\\
     P(A_H)-P_n(A_H)
    \end{array}\right] \Big|\; y_{1:n}
    \overset{\textup{d}}{\longrightarrow} \bm{Z}, \qquad \bm{Z} \sim \mathcal{N}_H(\bm0,\bm I),
\end{equation*}
for $n\rightarrow\infty$, where $\bm I$ denotes the identity matrix.
\end{proposition}

\begin{remark}
The results of Propositions~\ref{prop:asy_1} and~\ref{prop:asy_2} hold when $w_n = (\alpha+n)^{-\lambda}$ with $r_n = (2\lambda-1)n^{2\lambda-1}$.
\end{remark}

Replacement of the random matrix $\bm\Sigma$
with its estimate $\bm\Sigma_n$ is possible due to stable convergence.
Proposition~\ref{prop:asy_2} highlights the relationship between the posterior variance and one-step-ahead predictive updates.
This connection becomes particularly evident in the case $H=1$, where, for large $n$, the distribution of $P(A)\mid y_{1:n}$ can be approximated by $\mathcal{N}(P_n(A), \Sigma_nr_n^{-1})$, with $\Sigma_n=\sum_{y\in\mathcal{Y}}K_n(A\mid y)^2p_n(y) - P_n(A)^2$.
Therefore, as noted also by \cite{fortini_petrone_2020, fortini2024exchangeability} and \cite{fong2024asymptotics}, the asymptotic posterior variance of $P(A)$ is connected to the expected squared predictive update from $P_n(A)$ to $P_{n+1}(A)$ given $y_{1:n}$ with
\begin{equation}
    \Sigma_n r_n^{-1} 
    \approx \Sigma_n\sum_{k>n}w_k^2
    = \E\{[P_{n+1}(A)-P_n(A)]^2\mid y_{1:n}\}\sum_{k>n+1}w_k^2
    \label{eqn:asy_covariance}
\end{equation}
for $n$ large.
Equation~\eqref{eqn:asy_covariance} provides insights on the relationship between the learning rate of the predictive rule and the posterior variability, highlighting the crucial role played by the weights $w_n$: while weights that decay to zero quickly induce fast learning and convergence to the asymptotic exchangeability regime, small values of $w_n$ may result in poor learning and underestimation of the posterior variability.

Although the standard choice for the \textsc{dp} is 
$w_n=(\alpha+n)^{-1}$, it has been observed that such weights may decay to zero too quickly according to frequentist criteria, making them unappealing for \textsc{cid} modelling (see \citealp{fortini_petrone_2020}; \citealp{fong_holmes_2023}).
To address this, the literature has proposed alternatives with slower decay rates, such as 
$w_n=(\alpha+n)^{-2/3}$ \citep{martin2009asymptotic}, and 
$w_n=(2 - n^{-1})(n+1)^{-1}$ \citep{fong_holmes_2023}.
Alternatively, \cite{fortini_petrone_2020} suggest using adaptive weights of the form $w_n=(\alpha+n)^{-\lambda_n}$, where $\lambda_n=1$ for $n\le N_*$ to allow fast convergence, and $\lambda_n=3/4$ when $n>N_*$ induces slower decay once an approximate exchangeability regime is reached.
Although this split-sample approach is appealing, a smoother transition between the two regimes is desirable. For this reason, we consider the adaptive weights with
\begin{equation}
    \lambda_n = \lambda + (1 - \lambda)\exp\bigg\{-\frac{1}{N_*} n\bigg\},
    \label{eqn:adaptive}
\end{equation}
with $\lambda\in(1/2, 1]$ and $N_*>0$.
Using \eqref{eqn:adaptive}, the weights behave as $w_n \approx (\alpha+n)^{-1}$ for small $n$, smoothly transitioning to $w_n\approx(\alpha+n)^{-\lambda}$ as $n$ increases.
The transition rate is governed by $N_*$, which should be chosen such that the sequence is approximately in an exchangeable regime beyond that point.
A practical default is to set $N_*=N/2$, where $N$ denotes the length of each forward simulation in the predictive resampling algorithm.
Following \citet{fortini_petrone_2020}, we adopt $\lambda=3/4$ in our analyses.
The implications of these different weighting options for the \textsc{mad} sequences are explored in Section~\ref{sec:simulations}.

\subsection{Univariate illustrative example}
\label{sec:illustration}

We provide a simulated example to illustrate the \textsc{mad} sequence and compare it with the \textsc{dp}.
We generate $n=50$ data points from the mixture $0.4 \,\mathrm{Poisson}(y; 25) + 0.6\, \mathrm{Poisson}(y; 60)$. 
We consider two variants of \textsc{mad} sequences, having $w_n=(\alpha+n)^{-1}$ and $w_n= (\alpha+n)^{-\lambda_n}$, with $\lambda_n$ defined in \eqref{eqn:adaptive}, $\lambda = 3/4$ and $N_*=500$.
For updating the \textsc{mad} sequences, we choose a uniform base measure $P_0$ over the support, which we restrict to $\{0,\ldots,100\}$, and set $\alpha=1$.
We employ rounded Gaussian base kernels, and we recommend selecting $\sigma$ using a data-driven approach. Following \cite{fong_holmes_2023}, we rely on the notion of prequential log-likelihood introduced by \cite{dawid1984present}. Specifically, under the predictive construction \eqref{eqn:sequental_contruction}, the quantity $\sum_{i=1}^n \log p_{i-1}(y_i)$ can be interpreted as a log-likelihood function, making its maximization a natural criterion for selecting $\sigma$.
Alternatively, one can tune $\sigma$ using the methods available for \textsc{kde}, such as minimizing the prediction error in a validation set or by cross-validation.

\begin{figure}[t]
\centering
\includegraphics[width=0.95\textwidth]{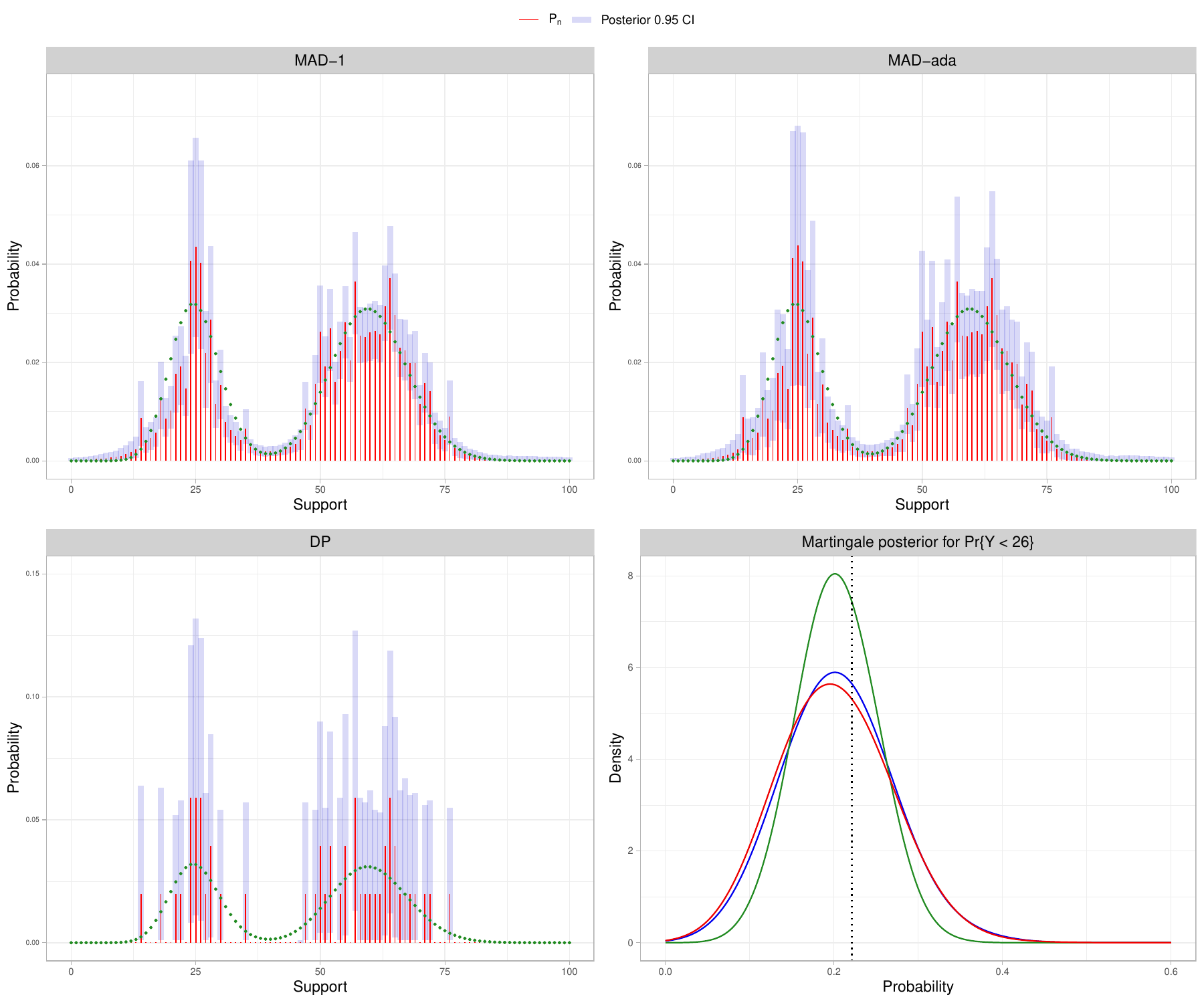}
\caption{\small Predictive distributions obtained using a \textsc{mad} sequence with $w_n=(\alpha+n)^{-1}$ (\textsc{mad}-1), with adaptive weights $w_n=(\alpha+n)^{-\lambda_n}$ (\textsc{mad}-ada), and with the \textsc{dp}.
The dotted line represents the probability mass function of the true data generator.
The bottom-right panel shows the martingale posterior distribution for the probability that $Y\le25$ using \textsc{mad}-1 (green), \textsc{mad}-ada (red) and \textsc{dp} (blue), with the dotted line representing the true value.}
\label{fig:univariate_example}
\end{figure}

Uncertainty quantification is carried out by predictive resampling $N=1000$ future observations starting from $P_n$ and obtaining $B=1000$ samples from the corresponding martingale posterior. By replacing exchangeability with the \textsc{cid} condition, \textsc{mad} predictive distributions depend on the ordering of $y_{1:n}$. To ensure permutation invariance, we compute $P_n$ averaging over $S$ permutations of the data.
Consistently with \cite{fong_holmes_2023}, we found that $S=10$ is sufficient in practice. Computations for different permutations can be performed in parallel. 
The prequential log-likelihood for selecting the optimal $\sigma$ is also averaged over the $S$ permutations, whereas the predictive resampling algorithm starts from the permutation averaged $P_n$.

Figure~\ref{fig:univariate_example} shows the predictive distributions obtained using the two variants of \textsc{mad} sequences and \textsc{dp}.
Both weight options provide similar \textsc{mad} predictive distributions and the improvement over the \textsc{dp} is substantial. The \textsc{mad} sequences effectively capture the two components of the mixture, assigning lower probabilities in the tails and in the region between the modes.
For each value $y\in\mathcal{Y}$, Figure~\ref{fig:univariate_example} also includes the 95\% credible interval for $p_N(y)$.
Additionally, the bottom-left panel of Figure~\ref{fig:univariate_example} presents the martingale posterior distribution for the quantity $\theta_N= P_N((-\infty, 25]) = \sum_{y\in\mathcal{Y}}\one(y \le 25)p_N(y)$, computed as described in Section~\ref{sec:martingale}.
As discussed in Section~\ref{sec:asymptotics}, the adaptive weights yield a more conservative uncertainty quantification compared to the standard choice $w_n=(\alpha+n)^{-1}$.
This results in wider 95\% credible intervals for $P_N(y)$ and a less concentrated martingale posterior for $\theta$.

\section{Multivariate modeling}
\label{sec:multiv_counts}

In this section, we consider \textsc{mad} sequences for modeling multivariate count data. Let $\bm{y}_n = (y_{n1},\ldots,y_{nd})\in \mathcal{Y} = \{0,1,\dots\}^d$ be a realization of the random vector $\bm{Y}_n = (Y_{n1},\dots,Y_{nd})$, for $n\ge1$ and $d\ge2$, and denote $\bm{y}_{1:n}=(\bm{y}_1^\top,\ldots,\bm{y}_n^\top)$.
Assuming the sequential generating mechanism \eqref{eqn:sequental_contruction} for $(\bm{Y}_n)_{n \ge 1}$, the multivariate \textsc{mad} predictive distribution has probability mass function
\begin{equation*}
    p_n(\y) = \mathds{P}(Y_{n+1,1}=y_1,\ldots,Y_{n+1,d}=y_d\mid \bm{y}_{1:n})
    = (1-w_n) p_{n-1}(\bm{y}) + w_n k_n(\bm{y}\mid\bm{y}_{n}),
\end{equation*}
where the multivariate \textsc{mh} kernel retains the same structure as in the univariate case, since the corresponding probability mass function is
\begin{equation*}
    k_n(\bm{y}\mid\bm{y}_n) = \gamma_n(\bm{y},\bm{y}_n)k_*(\bm{y}\mid\bm{y}_n) + \mathds{1}(\bm{y}=\bm{y}_n)\bigg[\sum_{\bm{z}\in\mathcal{Y}}\Big\{1-\gamma_n(\bm{z},\bm{y}_n)\Big\}k_*(\bm{z},\bm{y}_n)\bigg],
\end{equation*}
with
\begin{equation*}
    \gamma_n(\bm{y},\bm{y}_n) = \min\left\{ 1,\frac{p_{n-1}(\bm{y})k_*(\bm{y}_n\mid\bm{y})}{p_{n-1}(\bm{y}_n)k_*(\bm{y}\mid\bm{y}_n)}\right\}.
\end{equation*}
The above equations can be formally regarded as special cases of the general definition in \eqref{eqn:kernel_dp}, \eqref{eqn:mh_probability}, and \eqref{eqn:acc_probability}, but the multivariate nature of the data is emphasized using the bold notation. Any discrete multivariate distribution can be chosen as base kernel.
A simple and reasonable choice is a factorized base kernel having rounded Gaussian distribution components,
\begin{equation*}
    k_*(\bm{y}\mid\bm{y}_n,\bm\sigma) = \prod_{j=1}^d k_*(y_j\mid y_{nj},\sigma_j),
\end{equation*}
with $\bm\sigma = (\sigma_1,\ldots,\sigma_d)$. This factorized structure again parallels common practice in standard multivariate \textsc{kde}. More complex kernels that induce dependence between components can also be adopted.

\begin{remark}
\label{rem:multivariate} \textsc{mad} sequences are defined on a general countable space $\mathcal{Y}$, which includes $\mathcal{Y} = \{0, 1, \dots\}^d$ as a special case. Hence, the results of Theorem~\ref{th:cid} and Corollary~\ref{cor:cid}, as well as those of Propositions~\ref{prop:asy_1}--\ref{prop:asy_2}, hold for multivariate \textsc{mad} sequences. The multivariate case requires special attention, as the choice of the base kernel is less straightforward, and the practical implications deserve careful discussion.

\end{remark}

There are several advantages of multivariate \textsc{mad} sequences compared to existing approaches. First, although \textsc{dp} mixtures of multivariate rounded Gaussian kernels with non-diagonal covariance matrices \citep{canale2011bayesian} can be used to model multivariate count probability mass functions, \textsc{mcmc} computation can be inefficient. Conceptually, the predictive mixture model of \citet{fortini_petrone_2020} can address this issue, but in practice, the associated Newton algorithm requires the numerical evaluation of a multidimensional integral at each iteration. In contrast, employing \textsc{mh} kernels makes the extension to multivariate \textsc{mad} sequences both straightforward and computationally tractable, as only univariate Gaussian integrals are involved, for which efficient numerical methods are available. Furthermore, unlike the multivariate method proposed by \citet{fong_holmes_2023}, which depends on the order in which dimensions are updated, multivariate \textsc{mad} sequences are invariant to the dimension ordering. For the copula update, the \textsc{cid} assumption holds only under a fixed ordering of the variables, and different orderings lead to distinct predictive distributions. This sensitivity is unappealing in practice. In contrast, \textsc{mad} sequences assume that $(\bm Y_n)_{n \geq 1}$ is a \textsc{cid} sequence regardless of the variable ordering. In Appendix~\ref{appendix_copula}, we examine the dimension-ordering dependence of the copula update in a simulation study and compare the results with those obtained from \textsc{mad} sequences.

\subsection{Multivariate binary data}
\label{sec:multiv_binary}
Multivariate \textsc{mad} sequences can also be adapted to model multivariate binary data. Due to the flexibility of \textsc{mh} kernels, this extension is straightforward, requiring only the selection of an appropriate base kernel.
For $\bm{y}\in\{0,1\}^d$, we propose to employ the base kernel that has probability mass function
\begin{equation*}
    k_*(\y\mid\y_n) = \prod_{j=1}^d k_*(y_j\mid y_{nj}),
\end{equation*}
where
\begin{equation}
k_*(y_j\mid y_{nj}) = |y_{nj} - \delta|^{y_j}(1-|y_{nj}-\delta|)^{(1-y_j)} = 
\begin{cases}
\delta^{y_j}(1-\delta)^{1-y_j}\qquad\text{if}\;y_{nj} = 0\\
(1-\delta)^{y_j}\delta^{1-y_j}\qquad\text{if}\;y_{nj} = 1\\
\end{cases}
\label{eqn:kernel_binary}
\end{equation}
for $j=1,\ldots,d$ and $\delta\in[0,0.5]$.
An example of \eqref{eqn:kernel_binary} for different values of $\delta$ is shown in Appendix~\ref{appendix_figures}.
Roughly speaking, for each dimension $j$, if $y_{nj}=1$ the base kernel assigns probability $(1-\delta)\ge0.5$ to the event $Y_{n+1,j}=1$ and non-zero probability $\delta\le0.5$ to the event $Y_{n+1,j}=0$, and vice versa. For $\delta=0$, the \textsc{dp} update is recovered, while for $\delta=0.5$, a uniform base kernel is obtained. A practical default is to set $\delta=0.25$, or $\delta$ can be estimated using the prequential log-likelihood technique in Section~\ref{sec:rounded_gaussian}.

\subsection{Nonparametric regression and classification}

An important consequence of the multivariate approach is the possibility of using \textsc{mad} sequences for nonparametric regression and classification.  Let each covariate vector $\bm{x}_n \in \mathcal{X} \subseteq \{0,1,\dots\}^d$ be a realization of the random vector $\bm{X}_n$.  Then, we can first estimate the joint probability mass function $p_n(y,\bm{x})$ based on the joint \textsc{cid} sequence $(Y_1,\bm{X}_1),\dots,(Y_n, \bm{X}_n)$, and then compute the conditional predictive distribution $p_n(y \mid \bm{x})$, with $p_n(y \mid \bm{x}) = p_n(y,\bm{x}) / p_n(\bm{x})$. In this case, predictive resampling requires jointly sampling $(Y_{n+1}, \bm{X}_{n+1}) \sim p_n(y,\bm{x})$.  The response variable $Y_n$ can also be multivariate, possibly including both binary and count-valued components.  Moreover, the framework can incorporate continuous covariates by combining discrete \textsc{mh} kernels with the bivariate Gaussian copula update proposed by \citet{fong_holmes_2023}.  We provide simulation examples of count data regression and classification with binary covariates in Section~\ref{sec:simulations}.

\section{Simulation studies}
\label{sec:simulations}

In this section, we present two simulation studies comparing \textsc{mad} sequences with established machine learning methods for count data regression and classification. The studies assess both the predictive accuracy and the quality of uncertainty quantification. A separate simulation comparing \textsc{mad} sequences with the copula updating approach of \cite{fong_holmes_2023} is provided in Appendix~\ref{appendix_copula}.

\subsection{Count data regression with binary covariates}
\label{sec:sim_counts_regr}

Nonparametric methods are especially useful for capturing nonlinear relationships between variables.
To demonstrate the effectiveness of \textsc{mad} sequences in nonlinear regression settings, we conducted a simulation study comparing our approach with well-established machine learning methods, such as random forests and Bayesian additive regression trees (\textsc{bart}; \citealp{chipman2010bart}). 
Although machine learning methods typically excel with large sample sizes, we expect that \textsc{mad} sequences offer better performance in small-sample scenarios.
As additional benchmarks, we include the \textsc{dp} and the Poisson generalized linear model (\textsc{glm}). 
For each sample size $n \in \{40, 80\}$, we generate 50 simulated datasets. 
For each simulation, we include $10$ binary covariates sampled from independent binomial distributions with parameter $0.5$, and $Y_i\mid\x_i\sim\mathrm{Poi(e^{\eta_i})}$, where
\begin{equation*}
    \eta_i = 1 + \sqrt{\big|-0.5x_{i,1} + 1.5x_{i,2} +x_{i,3} +0.5x_{i,4} - 0.5x_{i,5}\big|} + \big(-0.7x_{i,6} +0.5x_{i,7} +0.7x_{i,8} -0.3x_{i,9} -0.3x_{i,10}\big)^2.
\end{equation*}
The study assesses the performance of each method by comparing the accuracy of out-of-sample predictions and uncertainty quantification.
For predictive accuracy, we compute the mean squared error (\textsc{mse}) obtained on $10^4$ new observations.
To evaluate uncertainty quantification accuracy, we consider the frequentist coverage of the 95\% credible intervals for $\E(Y\mid\x)$, calculated across $50$ simulations for each of the $2^{10} = 1024$ possible values for the covariate vector $\x$.
For random forests, we consider the predictive confidence intervals proposed by \citet{mentch2016quantifying}.

Since the posterior variability of \textsc{mad} sequences depends on the choice of the weights $w_n$, we compare four variants of \textsc{mad} sequences using the options discussed in Section~\ref{sec:asymptotics}:
$w_n=(\alpha+n)^{-1}$ (\textsc{mad-1}), $w_n = (\alpha+n)^{-2/3}$ (\textsc{mad-2/3}), the weights $w_n=(2-n^{-1})(n+1)^{-1}$ proposed by \cite{fong_holmes_2023} (\textsc{mad-dpm}), and the adaptive solution $w_n=(\alpha+n)^{\lambda_n}$, with $\lambda_n$ defined in Equation~\eqref{eqn:adaptive} (\textsc{mad-ada}).
We use a weakly informative uniform base measure over the support and set 
$\alpha=1$. For the adaptive weights, we set $\lambda=3/4$ and $N_*=500$.
Each \textsc{mad} sequence employs a rounded Gaussian base kernel for the response, with bandwidth selected by maximizing the prequential log-likelihood.
The hyperparameters for the binary base kernels for the covariates are set to the default value $0.25$.
To ensure permutation invariance, we average the predictive distribution and the prequential log-likelihood over $10$ random permutations of the data.
To compute the martingale posterior distributions we draw $B=200$ samples from $P_N$ using a predictive resampling algorithm with $N=n+1000$.
For point prediction with \textsc{mad} sequences, we use the population mean $\theta_n = \sum_{y\in\mathcal{Y}}y\, p_n(y\mid \x)$, while for uncertainty quantification, we employ the credible intervals obtained from the martingale posterior for the population mean, as described in Section~\ref{sec:martingale}.

\begin{table}[t]
\centering
\caption{\small Out-of-sample \textsc{mse} evaluated on $10^4$ new observations and median frequentist coverage for the 95\% credible intervals for $\E(Y\mid\x)$ for each of the 1024 possible values of the covariate vector $\x$.
The results are obtained by averaging $50$ simulated datasets. 
Bold values represent the lowest \textsc{mse} among the proposed methods and the lowest \textsc{mse} among the competitors, for each scenario.
Standard error and interquartile range are reported in brackets for \textsc{mse} and median coverage, respectively.}
\label{tab:sim_regr}
\vspace{0.2cm}
\footnotesize
\begin{tabular}{l c rl c rl cc rl c rl}
\toprule
Regression & & \multicolumn{5}{c}{$n=40$} & & & \multicolumn{5}{c}{$n=80$}\\
& & \multicolumn{2}{c}{\textsc{mse}} & & \multicolumn{2}{c}{Median coverage} & & & \multicolumn{2}{c}{\textsc{mse}} & & \multicolumn{2}{c}{Median coverage}\\
\midrule
\textsc{glm} & & $120.77$ & $[51.51]$ & &  $0.84$ & $[0.70, 0.90]$ & & & $94.93$ & $[8.37]$ & & $0.90$ & $[0.74, 0.96]$\\
\textsc{bart} & & $101.17$ & $[12.69]$ & &  $0.92$ & $[0.82, 0.96]$ & & & $\bm{74.17}$ & $[10.00]$ & & $0.92$ & $[0.84,0.96]$\\
\textsc{rf} & & $\bm{99.98}$ & $[7.45]$ & &  $1.00$ & $[0.92,1.00]$ & & & $87.75$ & $[6.53]$ & & $1.00$ & $[1.00,1.00]$\\
\textsc{dp} & & $1450.21$ & $[5.53]$ & &  $0.00$ & $[0.00, 0.00]$ & & & $1395.61$ & $[8.72]$ & & $0.00$ & $[0.00, 0.00]$\\
\midrule
\textsc{mad-1} & & $91.07$ & $[10.35]$ & &  $0.92$ & $[0.68,0.98]$ & & & $73.96$ & $[7.60]$ & & $0.88$ & $[0.58,0.98]$\\
\textsc{mad-2/3} & & $88.83$ & $[13.00]$ & &  $0.92$ & $[0.50,0.98]$ & & & $73.18$ & $[9.58]$ & & $0.94$ & $[0.56,1.00]$\\
\textsc{mad-dpm} & & $\bm{87.41}$ & $[12.36]$ & &  $0.90$ & $[0.52,0.98]$ & & & $\bm{72.07}$ & $[9.48]$ & & $0.92$ & $[0.54,1.00]$\\
\textsc{mad-ada} & & $90.61$ & $[10.28]$ & &  $0.96$ & $[0.77,1.00]$ & & & $73.45$ & $[7.69]$ & & $0.96$ & $[0.64,1.00]$\\
\bottomrule
\end{tabular}
\end{table}

\begin{table}[t]
\centering
\caption{\small \small Out-of-sample \textsc{auc} evaluated on $10^4$ new observations and median frequentist coverage for the 95\% credible intervals for $\E(Y\mid\x)$ for each of the 1024 possible values of the covariate vector $\x$.
The results are obtained by averaging $50$ simulated datasets. 
Bold values represent the highest \textsc{auc} among the proposed methods and the highest \textsc{auc} among the competitors, for each scenario.
Standard error and interquartile range are reported in brackets for \textsc{auc} and median coverage, respectively.}
\label{tab:sim_class}
\vspace{0.2cm}
\footnotesize
\begin{tabular}{l c rl c rl cc rl c rl}
\toprule
Classification & & \multicolumn{5}{c}{$n=150$} & & & \multicolumn{5}{c}{$n=300$}\\
& & \multicolumn{2}{c}{\textsc{auc}} & & \multicolumn{2}{c}{Median coverage} & & & \multicolumn{2}{c}{\textsc{auc}} & & \multicolumn{2}{c}{Median coverage}\\
\midrule
\textsc{glm} & & $0.796$ & $[0.014]$ & &  $0.18$ & $[0.06, 0.46]$ & & & $0.809$ & $[0.007]$ & & $0.00$ & $[0.00, 0.20]$\\
\textsc{bart} & & $0.863$ & $[0.026]$ & &  $0.18$ & $[0.00, 0.71]$ & & & $\bm{0.932}$ & $[0.009]$ & & $0.22$ & $[0.00,0.80]$\\
\textsc{rf} & & $\bm{0.882}$ & $[0.025]$ & &  $0.18$ & $[0.10,0.38]$ & & & $0.913$ & $[0.015]$ & & $0.70$ & $[0.60,0.84]$\\
\textsc{dp} & & $0.644$ & $[0.011]$ & &  $0.00$ & $[0.00,0.00]$ & & & $0.724$ & $[0.012]$ & & $0.00$ & $[0.00,0.02]$\\
\midrule
\textsc{mad-1} & & $0.873$ & $[0.014]$ & &  $0.44$ & $[0.08,0.73]$ & & & $0.899$ & $[0.008]$ & & $0.34$ & $[0.02,0.68]$\\
\textsc{mad-2/3} & & $0.869$ & $[0.015]$ & &  $0.86$ & $[0.72,0.96]$ & & & $0.899$ & $[0.009]$ & & $0.98$ & $[0.92,1.00]$\\
\textsc{mad-dpm} & & $0.872$ & $[0.014]$ & & $0.72$ & $[0.52,0.89]$ & & & $\bm{0.901}$ & $[0.008]$ & & $0.74$ & $[0.48,0.90]$\\
\textsc{mad-ada} & & $\bm{0.874}$ & $[0.014]$ & &  $0.82$ & $[0.62,0.94]$ & & & $0.900$ & $[0.008]$ & & $0.90$ & $[0.72,0.98]$\\
\bottomrule
\end{tabular}
\end{table}

The \textsc{mse} and coverage results obtained across the $50$ simulated datasets are reported in Table~\ref{tab:sim_regr}.
Since it does not allow the borrowing of information between nearby locations, the \textsc{dp} overfits the observed data, consequently leading to poor out-of-sample point predictions and inaccurate credible intervals.
In terms of predictive accuracy, as expected, \textsc{mad} sequences perform substantially better than 
the Poisson \textsc{glm}, which fails to capture nonlinear effects.
The four variants of \textsc{mad} sequences exhibit similar predictive performance, indicating that the choice of weights has little impact on the point estimate.
Furthermore, all \textsc{mad} sequences provide a better prediction accuracy than random forest and \textsc{bart} for both $n=40$ and $n=80$.

In terms of uncertainty quantification, the credible intervals produced by the \textsc{mad} sequence with adaptive weights are, on average, well calibrated for both sample sizes considered. 
In contrast, alternative weighting strategies tend to result in undercoverage. 
However, we observe substantial variability in empirical coverage in $2^{10} = 1024$ possible covariate configurations.
Although the Poisson \textsc{glm} and \textsc{bart} models exhibit less variability in coverage across covariate settings, they consistently fall below the nominal level, indicating systematic undercoverage.
Random forests produce overly conservative credible intervals, with coverage levels substantially exceeding $0.95$.

\subsection{Classification with binary covariates}
\label{sec:sim_class}

We also conducted a simulation study focused on binary classification from $10$ binary covariates. We generated $50$ simulated datasets for each sample size $n=\{150,300\}$. We sample $y_i\mid\bm{x}_i$ from ${\mathrm{Bin}\{1,\mathrm{logit}^{-1}(\eta_i)\}}$, with
\begin{equation*}
    \eta_i = -3 +2x_{i,1} -4x_{i,2} +3x_{i,3} -3x_{i,4} -3x_{i,5}x_{i,6} + \sqrt{\big|2x_{i,7} -3x_{i,8}\big|} + \big(3x_{i,9}-2x_{i,10}\big)^2,
\end{equation*}
$i=1,\ldots,n$, and draw $x_{i,j}\sim \mbox{Bernoulli}(\tau_j)$, $j=1,\ldots,10$, with $\bm\tau=(0.45, 0.65, 0.7, 0.4, 0.4, 0.6, 0.7, 0.3,\\0.55, 0.55)$ independently. We considered the variants of \textsc{mad} sequences described in Section~\ref{sec:sim_counts_regr} for count data regression. The hyperparameters for the binary base kernels for the covariates are set to the default value of $0.25$, while the response hyperparameter is selected by maximizing the prequential log-likelihood. We compare with \textsc{dp}, logistic \textsc{glm}, \textsc{bart}, and random forest. 
We evaluated out-of-sample predictive accuracy using the \textsc{roc} curve computed on $10^4$ new observations, and compared the corresponding area under the curve (\textsc{auc}).
Uncertainty quantification performance is evaluated as in Section~\ref{sec:sim_counts_regr}.
The results are reported in Table~\ref{tab:sim_class}. 

In terms of predictive accuracy, \textsc{mad} sequences outperform both the \textsc{dp} and logistic \textsc{glm} models, and are competitive with \textsc{bart} and random forests, particularly when $n = 150$.
For $n = 300$, \textsc{bart} and random forests produce slightly more accurate point predictions; however, the credible intervals produced by \textsc{mad} sequences are significantly better calibrated in terms of frequentist coverage.
In particular, although we do not claim optimal calibration guarantees, \textsc{mad} sequences with weights $w_n = (\alpha + n)^{-2/3}$, as well as those using adaptive weights, produce 95\% credible intervals with empirical coverage close to the nominal level, especially when $n = 300$.
In contrast, while providing good point predictions, the standard choice $w_n = (\alpha + n)^{-1}$ results in poorly calibrated credible intervals.

\section{Application: corvids abundance in Finland}
\label{sec:corvids}

We analyze data from an ecological monitoring study of birds in Finland \citep{lindstrom2015large}, focusing on the occurrence rates of corvids in year 2009.
Specifically, four species of the survey belong to the corvid family (\textit{Corvidae}): carrion crow (\textit{Corvus corone}), Eurasian magpie (\textit{Pica pica}), common raven (\textit{Corvus corax}), and Eurasian jay (\textit{Garrulus glandarius}). Corvids are non-migratory birds that occur throughout Finland. Monitoring their abundance is particularly important, as they can respond to habitat and climate changes more rapidly than most other species \citep{jokimaki2022long}. Moreover, their presence influences ecological balance, given their efficiency as predators and their ability to live in close associations with humans \citep{madden2015review}.
The study includes observations from $n=76$ different sites. For each observation, the data set includes environmental covariates that include habitat type (broad leaf forest, coniferous
forest, open habitat, urban habitat, or wetlands) and temperature in April-May (classified as cold, mild or warm if the registered average temperature is below $5$°C, between $5$° and $7.5$°C, or above $7.5$°C, respectively).

\begin{figure}[t]
\centering
\includegraphics[width=0.95\textwidth]{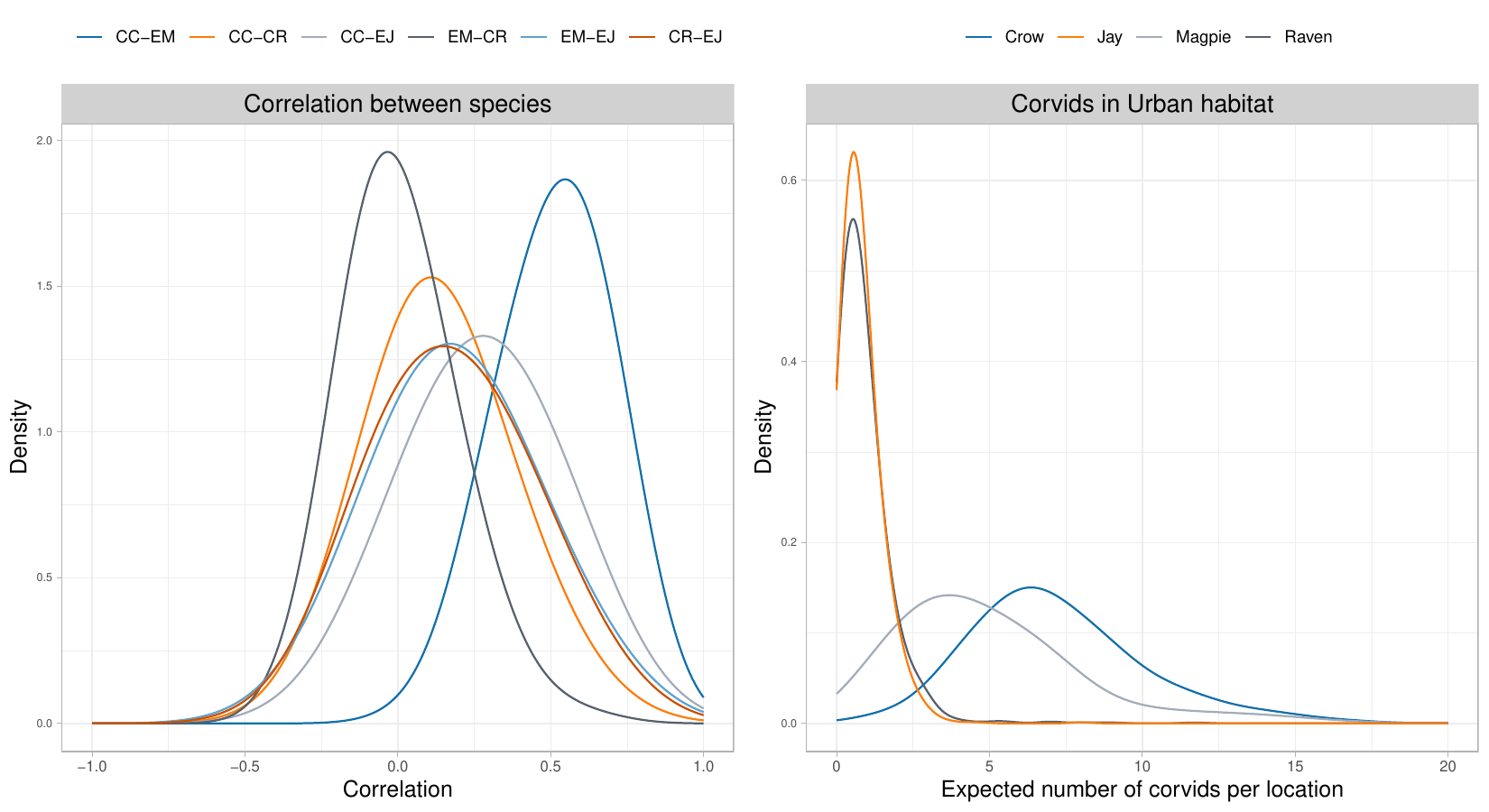}
\caption{\small
Left: posterior distribution for the correlation across species, with \textsc{cc}, \textsc{em}, \textsc{cr}, and \textsc{ej} denoting common crow, Eurasian magpie, common raven, and Eurasian jay, respectively.
Right: posterior distribution of the expected number of corvids observed at each urban habitat location.
}
\label{fig:corvids_1}
\end{figure}

Modeling bird species abundance data presents several challenges: count data are often zero inflated and overdispersed relative to the Poisson distribution, and the dependence structure between species is unknown a priori. \textsc{mad} sequences are particularly well suited for modeling such data, as their flexibility allows us to effectively address zero inflation and overdispersion while capturing across species dependence through the multivariate \textsc{mh} kernels described in Section~\ref{sec:multiv_counts}. We compare \textsc{mad} sequences with three widely used methods in ecological modeling: hierarchical modeling of species communities (\textsc{hmsc}; \citealp{ovaskainen2016using, tikhonov2017using, ovaskainen2020joint}), generalized joint attribute models (\textsc{gjam}s; \citealp{clark2017generalized}), and generalized linear latent variable models (\textsc{gllvm}s; \citealp{skrondal2004generalized}).
These models capture species dependence by incorporating latent variables in the linear predictor, and \textsc{hmsc} additionally includes site-specific random intercepts. To account for overdispersion and zero inflation, we consider negative binomial (\textsc{gllvm-nb}) and zero-inflated negative binomial (\textsc{gllvm-zinb}) distributions for the response in \textsc{gllvm}s.

For the \textsc{mad} sequence, we choose the adaptive weights $w_n=(\alpha+n)^{-\lambda_n}$, with $\lambda_n$ defined in Equation~\eqref{eqn:adaptive}. We use a weakly informative uniform base measure setting $\alpha=1$. As default choice we set $\lambda=3/4$ and $N_*=500$, as it also proved to be effective in the simulation study of Section~\ref{sec:sim_counts_regr}.
Then, we select hyperparameters by minimizing the prequential log-likelihood, obtaining $\bm\delta = (0.45, 0.41)$ for temperature and habitat, respectively, and $\bm\sigma = (3.81, 1.53, 1.12, 1.06)$ for the response. For permutation invariance, we average the \textsc{mad} predictive and the corresponding prequential log-likelihood over $10$ permutations. 
Computing the \textsc{mad} predictive distribution averaged over permutations takes $25.68$ seconds on a 3.2 GHz 8-Core Apple M1 CPU.
Uncertainty quantification is carried out using a predictive resampling algorithm with $N=n+1000$ and $B=1000$. To compare methods, we evaluate out-of-sample predictive accuracy using predictive mean squared error (\textsc{mse}) computed with respect to the data obtained at the same locations for year 2010. 
As a point estimate for \textsc{mad} sequences we consider the population mean computed at time $n$.
The \textsc{mad} sequence had the best predictive performance ($\textsc{mse}= 3.40$), followed by \textsc{gjam} ($3.78$), \textsc{gllvm-zinb} ($4.69$), \textsc{hmsc} ($4.70$), \textsc{gllvm-nb} ($4.72$).
Moreover, as shown in Appendix~\ref{appendix_figures}, the \textsc{mad} sequence performs well in capturing the dependence structure between species. In contrast, popular parametric latent variable models built from elaborations of hierarchical \textsc{glm}s struggle to recover the empirical dependence structure. 

\begin{figure}[t]
\centering
\includegraphics[width=\textwidth]{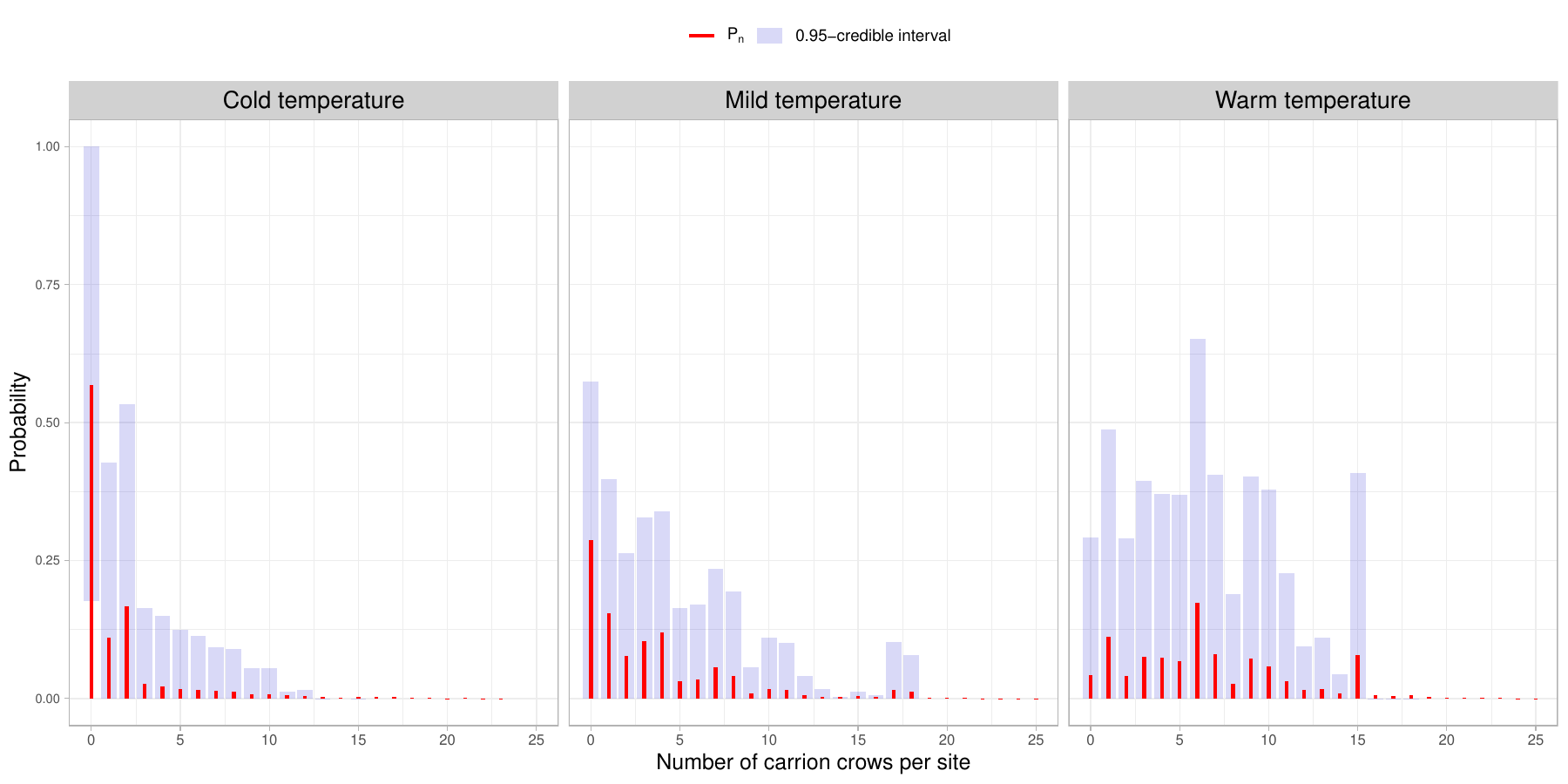}
\caption{\small \textsc{mad} predictive distribution of carrion crows across temperatures (red) with corresponding 0.95 credible interval (blue).}
\label{fig:corvids_2}
\end{figure}

The results obtained using \textsc{mad} sequences provide valuable insights. 
The right panel of Figure~\ref{fig:corvids_1} presents the posterior distribution of the population mean for each species in each urban location, obtained as discussed in Section~\ref{sec:martingale}.
The figure suggests that there are more crows and magpies (95\% credible intervals $[3.15, 11.28]$ and $[0.70, 8.81]$, respectively) in urban habitats compared to ravens and jays (95\% c.i. $[0.06, 1.61]$ and $[0.11, 1.57]$, respectively).
This aligns with \cite{matsyura2016spatial} and \cite{jokimaki2022long}, which indicate that urban settlements serve as stable wintering environments for magpies and crows.
As shown in the left panel of Figure~\ref{fig:corvids_1}, there is a positive dependence between these two species (posterior mean for correlation is $0.45$, with 95\% c.i. $[0.24,0.69]$), probably due to shared habitat preferences. 
A similar result was observed by \cite{dupak2023interactions} for hooded crows (a subspecies of carrion crows) and magpies in Ukraine.
Moreover, urban environments are also the preferred habitat for magpies (as reported in Appendix~\ref{appendix_figures}), probably because urbanization buffers seasonal variations in climate conditions and food availability \citep{benmazouz2021corvids}, facilitating magpies adaptation.
However, within the urban environment, Eurasian magpies avoid densely populated areas \citep{jokimaki2017urbanization}.
The differences in the estimated distributions at different temperatures represented in Figure~\ref{fig:corvids_2} suggest that carrion crows prefer areas with higher temperature, which correspond to the Southern Finland region (see Appendix~\ref{appendix_figures}).
Finally, consistent with \cite{madge1994crows}, our results indicate that the distributions of common ravens and Eurasian jays are homogeneous across different habitats and temperature conditions. 

\section{Discussion}

Our proposed \textsc{mad} sequences provide a flexible and effective approach for nonparametric modeling of discrete data distributions. We have shown that these sequences implicitly define a prior over the data-generating mechanism which, unlike the \textsc{dp}, can induce positive correlations between nearby values in the support, making \textsc{mad} sequences appealing in practical applications.

As discussed in Section~\ref{sec:asymptotics}, the weight sequence $(w_n)_{n \ge 1}$ plays a dual role, governing both the posterior variability induced by the \textsc{mad} prior and the learning rate of the associated predictive rule. Although various choices for $w_n$ have been proposed---some of which perform well empirically, as illustrated in Section~\ref{sec:simulations}---none currently guarantees well-calibrated frequentist coverage of posterior credible intervals. While this issue can be addressed in parametric settings \citep{fong2021conformal, fong2024asymptotics}, it remains an open challenge in the context of Bayesian nonparametric predictive inference, representing an important area for future research \citep{fortini2024exchangeability}.

Although \textsc{mad} sequences have proven effective for modeling multivariate count data, including in regression settings, their performance is likely to degrade in high-dimensional cases due to the curse of dimensionality. A promising direction for future work is to incorporate \textsc{mad} sequences into more structured models---for example, by using them to define priors over latent parameters. See \citet{Airoldi2014} for a precursor to these ideas. In particular, \textsc{cid} sequences could be employed to model low-dimensional latent variables underlying high-dimensional observations. Interesting next steps include establishing theoretical guarantees for these approaches and ensuring their computational scalability.

\section*{Acknowledgements}
This research was partially supported by the National Institutes of Health (grant ID R01ES035625), by the European Research Council under the European Union’s Horizon 2020 research and 
innovation programme (grant agreement No 856506), by the National Science Foundation (NSF IIS-2426762), and by the Office of Naval Research (N00014-21-1-2510). T.R. is supported by the European Union – Next Generation EU funds, component M4C2, investment 1.1., PRIN-PNRR 2022 (P2022H5WZ9).

\clearpage
\appendix

\section{Proofs}
\label{appendix_proofs}

\subsection{Preliminary results}

In this Section, we provide useful technical results for some of the later proofs, in particular for the asymptotic results. The following theorems from \cite{berti2013exchangeable} are employed in Lemma~\ref{lemma_1} for proving the convergence in total variation distance of the sequence of \textsc{mad} predictive distributions and they are included to make the paper self-contained.

\begin{theorem}[Theorem 1 in \citealp{berti2013exchangeable}]
\label{th:berti_2013_1}
Let $\lambda$ be a $\sigma$-finite measure on a Polish space $S$.
Let the sequence $(Y_n)_{n\ge1}$ be \textsc{cid} with
\begin{equation*}
    P_n(\cdot) = \PP\{Y_{n+1}\in\cdot\mid Y_1,\ldots,Y_n\}\longrightarrow P(\cdot)
\end{equation*}
weakly $\PP$-almost surely.
Then, $P\ll\lambda$ $\PP$-a.s. if and only if $||P_n-P||_{TV}\rightarrow0$ $\PP$-a.s. and $P_n\ll\lambda$ for all $n\ge1$.
\end{theorem}

\begin{theorem}[Theorem 4 in \citealp{berti2013exchangeable}]
\label{th:berti_2013_4}
Suppose $(Y_n)_{n\ge1}$ is \textsc{cid} and $P_n\ll\lambda$ for every $n\ge1$, with $p_n$ the density of $P_n$ with respect to $\lambda$.
Then, $P\ll\lambda$ if and only if, for every compact $K$ such that $\lambda(K)<\infty$, $p_n$ is a function on $S$ uniformly integrable with respect to $\lambda_K$ $\PP$-almost surely, where $\lambda_K(\cdot)=\lambda(\cdot\cap K)$ is the restriction of $\lambda$ to $K$.

In particular, $P\ll\lambda$ $\PP$-a.s. if, for every compact $K$ such that $\lambda(K)<\infty$, there exists $d>1$ such that
\begin{equation}
\label{eqn:condition_berti}
\sup_n\int_K p_n(y)^d\dd\lambda(y) < \infty
\end{equation}
$\PP$-a.s.
Moreover, for condition \eqref{eqn:condition_berti} to be true, it suffices that
\begin{equation*}
    \sup_n\E\left\{\int_Kp_n(y)^d\dd\lambda(y)\right\} < \infty
\end{equation*}
$\PP$-a.s.
\end{theorem}

\begin{lemma}
\label{lemma_1}
Let $(P_n)_{n\ge1}$ be a sequence of \textsc{mad} predictive distributions and $P$ its limit distribution.
Then, $P_n\rightarrow P$ in total variation $\PP$-almost surely as $n\rightarrow\infty$.
\end{lemma}
\begin{proof}
From Theorem~\ref{th:berti_2013_1}, $P_n\rightarrow P$ in total variation $\PP$-a.s. if $P$ is absolutely continuous with respect to the counting measure $\lambda$.
From Theorem~\ref{th:berti_2013_4}, $P$ is absolutely continuous with respect to $\lambda$ if there exists a $d>1$ such that
\begin{equation*}
    \sup_n \E\left\{\int_Kp_n(y)^d\dd\lambda(y)\right\}<\infty
\end{equation*}
for every compact $K$ such that $\lambda(K)<\infty$, where the expectation is taken with respect to $Y_1,\ldots,Y_n$.

In our case, since $p_n(y)\in[0,1]$ for every $y\in\mathcal{Y}$ and every $n\ge0$, we have that
\begin{equation*}
    \sup_n\E\left\{\int_Kp_n(y)^d\dd\lambda(y)\right\} 
    = \sup_n\E\left\{\sum_{y\in\mathcal{Y}}p_n(y)^d\delta_K(y)\right\}
    \le |K|
    < \infty
\end{equation*}
by definition of $K$.
\end{proof}

Lemma~\ref{lemma_2} and~\ref{lemma_3} provide the convergence of the \textsc{mh} kernel and its squared expectation, respectively. These results are useful for the proofs of Proposition~\ref{prop:asy_1} and~\ref{prop:asy_2}. 

\begin{lemma}
\label{lemma_2}
For every $y\in\mathcal{Y}$ and every measurable set $A$, 
\begin{equation*}
    K_n(A\mid y) \longrightarrow K(A\mid y)
\end{equation*}
$\PP$-a.s. as $n\rightarrow\infty$, where $K(A\mid y) = \sum_{x\in A}k(x\mid y)$, with
\begin{align*}
    k(x\mid y) & = \gamma(x,y)k_*(x\mid y) + \one(x = y)\left[1-\sum_{z\in\mathcal{Y}}\gamma(z,y)k_*(z,y)\right],\\
    \gamma(x,y) & = \min\left\{1,\frac{p(x)k_*(y\mid x)}{p(y)k_*(x\mid y)}\right\}.
\end{align*}
\end{lemma}
\begin{proof}
By Lemma~\ref{lemma_1}, $\lim_{n\rightarrow\infty}p_n(y) = p(y)$, $\PP$-a.s., for every $y\in\mathcal{Y}$.
Then, for every $x,y\in\mathcal{Y}$, 
\begin{equation*}
\lim_{n\rightarrow\infty}\gamma_n(x,y) 
= \lim_{n\rightarrow\infty}\min\left\{1,\frac{p_n(x)k_*(y\mid x)}{p_n(y)k_*(x\mid y)}\right\}
= \min\left\{1,\frac{p(x)k_*(y\mid x)}{p(y)k_*(x\mid y)}\right\}
=  \gamma(x,y)
\end{equation*}
$\PP$-almost surely.
It follows that
\begin{align*}
\lim_{n\rightarrow\infty}K_n(A\mid y) & = \lim_{n\rightarrow\infty}\sum_{x\in A}k_n(x\mid y)\\
& = \sum_{x\in A} \lim_{n\rightarrow\infty}\left\{\gamma_n(x,y)k_*(x\mid y) + \one(x=y)\left[1-\sum_{z\in\mathcal{Y}}\gamma_n(z,y)k_*(z\mid y)\right]\right\}\\
& = \sum_{x\in A} \left\{\gamma(x,y)k_*(x\mid y) + \one(x=y)\left[1-\sum_{z\in\mathcal{Y}}\gamma(z,y)k_*(z\mid y)\right]\right\}\\
& = K(A\mid y).
\end{align*}
\end{proof}

\begin{lemma}
\label{lemma_3}
For every $H\ge1$, all measurable sets $A_1,\ldots,A_H$, and every vector $\bm c = (c_1,\ldots,c_H)$ such that $||\bm c|| = 1$,
\begin{equation*}
    \E\Bigg\{\Bigg[\sum_{j=1}^Hc_jK_n(A_j\mid Y_{n+1})\Bigg]\Bigg[\sum_{r=1}^Hc_rK_n(A_r\mid Y_{n+1})\Bigg]\mid Y_{1:n}\Bigg\}
    \longrightarrow
    \sum_{j,r=1}^H c_j c_r\Bigg[\sum_{y\in\mathcal{Y}}K(A_j\mid y)K(A_r\mid y)p(y)\Bigg]
\end{equation*}
$\PP$-almost surely for $n\rightarrow\infty$.
\end{lemma}
\begin{proof}
We can write
\begin{align*}
\E\Bigg\{\Bigg[\sum_{j=1}^Hc_jK_n(A_j\mid Y_{n+1})\Bigg] & \left[\sum_{r=1}^Hc_rK_n(A_r\mid Y_{n+1})\right]\mid Y_{1:n}\Bigg\} \\
& = \sum_{y\in\mathcal{Y}}\Bigg[\sum_{j=1}^Hc_jK_n(A_j\mid y)\Bigg]\left[\sum_{r=1}^Hc_rK_n(A_r\mid y)\right]p_n(y)\\
& = \sum_{j=1}^H\sum_{r=1}^H c_j c_r \Bigg[\sum_{y\in\mathcal{Y}}K_n(A_j\mid y)K_n(A_r\mid y)p_n(y)\Bigg].
\end{align*}
Then, we have that
\begin{align*}
\Bigg|&\sum_{y\in\mathcal{Y}}K_n(A_j\mid y)K_n(A_r\mid y)p_n(y) - \sum_{y\in\mathcal{Y}}K(A_j\mid y)K(A_r\mid y)p(y)\Bigg|\\
& = \Bigg|\sum_{y\in\mathcal{Y}}K_n(A_j\mid y)K_n(A_r\mid y)p_n(y) - \sum_{y\in\mathcal{Y}}K_n(A_j\mid y)K_n(A_r\mid y)p(y)\\
& \hspace{5.5cm}+\sum_{y\in\mathcal{Y}}K_n(A_j\mid y)K_n(A_r\mid y)p(y) - \sum_{y\in\mathcal{Y}}K(A_j\mid y)K(A_r\mid y)p(y)\Bigg|\\
& \le \Bigg|\sum_{y\in\mathcal{Y}}K_n(A_j\mid y)K_n(A_r\mid y)p_n(y) - \sum_{y\in\mathcal{Y}}K_n(A_j\mid y)K_n(A_r\mid y)p(y)\Bigg|\\
& \hspace{5.5cm}+\Bigg|\sum_{y\in\mathcal{Y}}K_n(A_j\mid y)K_n(A_r\mid y)p(y) - \sum_{y\in\mathcal{Y}}K(A_j\mid y)K(A_r\mid y)p(y)\Bigg|\\
& \le \sum_{y\in\mathcal{Y}}K_n(A_j\mid y)K_n(A_r\mid y)\big|p_n(y) - p(y)\big|+\sum_{y\in\mathcal{Y}}\bigg|K_n(A_j\mid y)K_n(A_r\mid y) - K(A_j\mid y)K(A_r\mid y)\bigg|\;p(y)\\
& \le \sum_{y\in\mathcal{Y}}\big|p_n(y) - p(y)\big|+\sum_{y\in\mathcal{Y}}\bigg|K_n(A_j\mid y)K_n(A_r\mid y) - K(A_j\mid y)K(A_r\mid y)\bigg|\;p(y)\\
& \longrightarrow 0 \qquad \PP\mathrm{-a.s}\;\;\mathrm{for}\;\;n\rightarrow\infty,
\end{align*}
since the first term converges to zero $\PP$-a.s. because $P_n\rightarrow P$ in $L^1$ by Lemma~\ref{lemma_1} and the second term converges to zero by the dominated convergence theorem because $K_n(A\mid y)\longrightarrow K(A\mid y)$ by Lemma~\ref{lemma_2}.
\end{proof}

Finally, the following Theorems provide the key results for the proof of Proposition~\ref{prop:asy_1}. 

\begin{theorem}[Theorem A1 in \citealp{crimaldi2009almost}]
\label{th:crimaldi_2009}
On $(\Omega, \mathcal{F},\PP)$, for each $n\ge1$, let $(\mathcal{F}_{n,j})_{j\ge0}$ be a filtration and $(M_{n,j})_{j\ge0}$ a real martingale with respect to $(\mathcal{F}_{n,j})_{j\ge0}$, with $M_{n,0} = 0$, which converges in $L^1$ to a random variable $M_{n,\infty}$.
Set
\begin{equation*}
Z_{n,j}:=M_{n,j}-M_{n,j-1}\quad\mathrm{for}\;j\ge1,\qquad
U_n := \sum_{j\ge1}Z_{n,j}^2,\qquad
Z_n^* := \sup_{j\ge1}|Z_{n,j}|.
\end{equation*}
Further, let $(c_n)_{n\ge1}$ be a sequence of strictly positive integers such that $c_nZ_n^*\rightarrow0$ $\PP$-a.s. and let $\mathcal{U}$ be a sub-sigma field which is asymptotic for the conditioning system $\mathcal{G}$ defined by $\mathcal{G}_n=\mathcal{F}_{n,c_n}$.

Assume that the sequence $(Z_n^*)_{n\ge1}$ is dominated in $L^1$ and that the sequence $(U_n)_{n\ge1}$ converges almost surely to a positive real random variable $U$ which is measurable with respect to $\mathcal{U}$.
Then, with respect to the conditioning system $\mathcal{G}$, the sequence $(M_{n,\infty})_{n\ge1}$ converges to a Gaussian kernel $\mathcal{N}(0,U)$ in the sense of almost sure conditional convergence.
\end{theorem}

\begin{theorem}[Lemma 4.1 in \citealp{crimaldi2016fluctuation}]
\label{th:crimaldi_2016}
Let $\mathcal{G}=(\mathcal{G}_n)_{n\ge1}$ be a filtration and $(Z_n)_{n\ge1}$ be a $\mathcal{G}$-adapted sequence of real random variables such that $\E\{Z_n\mid\mathcal{G}_{n-1}\}\rightarrow Z$, $\PP$-a.s. for some random variable $Z$.
Moreover, let $(a_n)_{n\ge1}$ and $(b_n)_{n\ge1}$ be two positive sequences of strictly positive real numbers such that
\begin{equation*}
    b_n \uparrow +\infty, \qquad
    \sum_{n\ge1}(a_n^2b_n^2)^{-1}\E\{Z_n\} < \infty.
\end{equation*}
Then we have:
\begin{itemize}
\item[i)] if $b_n^{-1}\sum_{m=1}^n a_m^{-1}\rightarrow \zeta$ for some constant $\zeta$, then
\begin{equation*}
    \frac{1}{b_n}\sum_{m=1}^n \frac{Z_m}{a_m}\overset{a.s.}{\rightarrow}\zeta Z;
\end{equation*}
\item[ii)] if $b_n\sum_{m\ge n}(a_mb_m^2)^{-1}\rightarrow \zeta$ for some constant $\zeta$, then
\begin{equation*}
    b_n\sum_{m>n} \frac{Z_m}{a_mb_m^2}\overset{a.s.}{\rightarrow}\zeta Z;
\end{equation*}
\end{itemize}
\end{theorem}

\subsection{Proof of Theorem~\ref{th:cid}}
To prove that the \textsc{mad} sequence $(Y_n)_{n\ge1}$ is \textsc{cid}, it is sufficient to show that, for all $n\ge0$ and all measurable sets $A$,
\begin{equation*}
    \mathds{P}\{Y_{n+2}\in A\mid y_{1:n}\} = \mathds{P}\{Y_{n+1}\in A\mid y_{1:n}\},
\end{equation*}
which is equivalent to write 
\begin{equation}
    \mathds{P}\{Y_{n+2}\in A\mid y_{1:n}\} \equiv \mathds{E}\{P_{n+1}(A)\mid y_{1:n}\} = P_n(A) \equiv \mathds{P}\{Y_{n+1}\in A\mid y_{1:n}\}.
    \label{eqn:proof_cid}
\end{equation}
By rewriting
\begin{align*}
    \mathds{E}\{P_{n+1}(A)\mid y_{1:n}\} & 
    = \mathds{E}\big\{(1-w_{n+1})P_n(A) + w_{n+1}K_{n}(A\mid Y_{n+1})\mid y_{1:n}\big\}\\
    & = (1-w_{n+1})P_n(A) + w_{n+1}\mathds{E}\{K_n(A\mid Y_{n+1})\mid y_{1:n}\},
\end{align*}
we see that \eqref{eqn:proof_cid} is true if $\mathds{E}\{K_n(A\mid Y_{n+1})\mid y_{1:n}\} = P_n(A)$.
This is true because
\begin{equation*}
    \mathds{E}\{K_n(A\mid Y_{n+1})\mid y_{1:n}\}
    = \sum_{t\in A}\sum_{z\in\mathcal{Y}} k_n(t\mid z)p_n(z)
    \overset{(*)}{=} \sum_{t\in A}\sum_{z\in\mathcal{Y}} k_n(z\mid t)p_n(t)
    = \sum_{t\in A} p_n(t)
    = P_n(A),
\end{equation*}
for every $A$ and $n\ge0$, where equivalence $(*)$ is a consequence of the detailed balance condition of \textsc{mh} kernels, that is
\begin{equation*}
    k_n(z\mid y) p_n(y) = k_n(y\mid z)p_n(z),
\end{equation*}
hence
\begin{equation*}
    \sum_{z\in\mathcal{Y}} k_n(y\mid z)p_n(z) = \sum_{z\in\mathcal{Y}} k_n(z\mid y)p_n(y) = p_n(y).
\end{equation*}

\subsection{Proof of Corollary~\ref{cor:cid}}
\begin{itemize}
\item[$(a)$] This is a consequence of \citet[][Theorem 2.5]{berti2004limit}.
\item[$(b)$] This is a consequence of \citet[][Lemma 2.4, Theorem 2.2]{berti2004limit}.
\item[$(c)$] Define $\mathcal{P}$ as the space of probability functions on $\mathcal{Y}$.
Let $\Pi(\cdot)$ and $\Pi(\cdot\mid y_{1:n})$ denote the prior and posterior law of $P$, respectively.
Then,
\begin{align*}
    \E(\theta) = \E\{P(f)\} = \E\bigg\{\sum_{y\in\mathcal{Y}}f(y)p(y)\bigg\} 
    = \sum_{y\in\mathcal{Y}}f(y)\int_{\mathcal{P}}p(y)\Pi(\dd P)
    = \sum_{y\in\mathcal{Y}}f(y)p_0(y)
    = P_0(f) = \theta_0
\end{align*}
and
\begin{align*}
    \E(\theta\mid y_{1:n}) = \E\{P(f)\mid y_{1:n}\} & = \E\bigg\{\sum_{y\in\mathcal{Y}}f(y)p(y)\mid y_{1:n}\bigg\}\\
    & = \sum_{y\in\mathcal{Y}}f(y)\int_{\mathcal{P}}p(y)\Pi(\dd P\mid y_{1:n})\\
    & = \sum_{y\in\mathcal{Y}}f(y)p_n(y)\\
    & = P_n(f) = \theta_n,
\end{align*}
where $\int_{\mathcal{P}}p(y)\Pi(\dd P) = p_0(y)$ and $\int_{\mathcal{P}}p(y)\Pi(\dd P\mid y_{1:n}) = p_n(y)$ follows because the sequence $(Y_n)_{n\ge1}$ is \textsc{cid}.
\end{itemize}

\subsection{Proof of Proposition~\ref{prop:asy_1}}
The proof specializes the proof of \citet[Theorem~7]{fortini_petrone_2020}, which is based on Theorem~\ref{th:crimaldi_2009} and Theorem~\ref{th:crimaldi_2016}.
By Cramer-Wold device, it is sufficient to show that, for every vector $\bm c=(c_1,\ldots,c_H$) such that $||\bm c||=1$ and every $t\in\R$,
\begin{equation}
\PP\left\{\sqrt{r_n}\sum_{h=1}^Hc_h\left[P(A_h)-P_n(A_h)\right]\le t\mid Y_{1:n}\right\}\longrightarrow\Phi\left\{(\bm c^\top\bm\Sigma\bm c)^{-1/2}t\right\}
\label{eqn:prop_1_goal}
\end{equation}
$\PP$-almost surely, where $\Phi(\cdot)$ denotes the standard normal cumulative distribution function.

The sequence $(\sum_{h=1}^H c_h P_n(A_h))_{n\ge1}$ is a bounded martingale coverging to $\sum_{h=1}^H c_h P(A_h)$ $\PP$-a.s. in $L^1$, as consequence of Lemma~\ref{lemma_1}.
Let
\begin{equation*}
    M_{n,j} = \sqrt{r_n}\left[\sum_{h=1}^H c_hP_n(A_h) - \sum_{h=1}^H c_h P_{n-j+1}(A_h)\right],
    \qquad \mathcal{F}_{n,j} = Y_{1:n+j-1}
\end{equation*}
for $j\ge1$ and $M_{n,0} = 0$, $\mathcal{F}_{n,0} = Y_{1:n}$.
$(M_{n,j})_{j\ge0}$ is a zero-mean martingale with respect to the filtration $(\mathcal{F}_{n,j})_{j\ge1}$ under $\PP$.
Let
\begin{equation*}
Z_{n,j} := M_{n,j} - M_{n,j-1} = \sqrt{r_n} w_{n+j-1} T_{n+j-1},\quad
U_n := \sum_{j\ge1} Z_{n,j}^2 = r_n \sum_{j\ge1}  w_{n+j-1}^2 T_{n+j-1}^2,\quad
Z_n^* := \sup_{j\ge1}|Z_{n,j}|, 
\end{equation*}
with
\begin{equation*}
T_{n+j-1} = \sum_{h=1}^H c_h \left[P_{n+j-2}(A_h) - K_{n+j-2}(A_h\mid Y_{n+j-1})\right].
\end{equation*}
Equation~\eqref{eqn:prop_1_goal} follows from Theorem~\ref{th:crimaldi_2009} if we can show that $(Z_n^*)_{n\ge1}$ is dominated in $L^1$ and $(U_n)_{n\ge1}$ converges $\PP$-a.s. to $\sum_{h,r =1}^H c_hc_r\Sigma_{hr}$.
$(Z_n^*)_{n\ge1}$ is dominated in $L^1$ because
\begin{equation*}
    \sup_{j\ge1}|Z_{n,j}| 
    = \sqrt{r_n}\sup_{j\ge1}|w_{n+j-1}||T_{n+j-1}| 
    \le \sqrt{r_n}\sup_{j\ge1}|w_{n+j-1}| \longrightarrow 0
\end{equation*}
by assumption.
To prove that $U_n$ converges $\PP$-a.s. to $\sum_{h,r =1}^H c_hc_r\Sigma_{hr}$ we employ part $ii)$ of Theorem~\ref{th:crimaldi_2016}, with
\begin{equation*}
    b_{m+1} = r_m, \qquad a_m = (b_m^2 w_m^2)^{-1}
\end{equation*}
for $m\ge1$ and $b_1=r_1$.
Then, we can write
\begin{equation*}
    U_n = b_{n+1} \sum_{s \ge n+1}(a_sb_s^2)^{-1}T_s^2.
\end{equation*}
By Lemma~\ref{lemma_3}, we have that, $\PP$-a.s.,
\begin{align*}
    \E\big\{T_s^2\mid Y_{1:s-1}\big\} = & \E\Bigg\{ \Bigg[\sum_{h=1}^H c_h P_{s-1}(A_h) - \sum_{h=1}^H c_h K_{s-1}(A_h\mid Y_s)\Bigg]\\
    & \hspace{5cm} \times \Bigg[\sum_{r=1}^H c_r P_{s-1}(A_r) - \sum_{r=1}^H c_r K_{s-1}(A_r\mid Y_s)\Bigg] \mid Y_{1:s-1} \Bigg\}\\
    = & \E\Bigg\{\sum_{h=1}^H c_h P_{s-1}(A_h) \sum_{r=1}^H c_r P_{s-1}(A_r)
    +\sum_{h=1}^H c_h K_{s-1}(A_h\mid Y_s) \sum_{r=1}^H c_r K_{s-1}(A_r\mid Y_s)\\
    & \hspace{0.5cm} -\sum_{h=1}^H c_h P_{s-1}(A_h) \sum_{r=1}^H c_r K_{s-1}(A_r\mid Y_s)
    - \sum_{h=1}^H c_h K_{s-1}(A_h\mid Y_s) \sum_{r=1}^H c_r P_{s-1}(A_r) \Big| Y_{1:s-1}\Bigg\}\\
    = & \E\Bigg\{\sum_{h=1}^H c_h K_{s-1}(A_h\mid Y_s)\sum_{r=1}^H K_{s-1}(A_r\mid Y_s)\mid Y_{1:s-1}\Bigg\} - \sum_{h=1}^H c_h P_{s-1}(A_h) \sum_{r=1}^H c_r P_{s-1}(A_r)\\
    = & \sum_{h,r=1}^H c_h c_r \left[\sum_{y\in\mathcal{Y}}K_{s-1}(A_h\mid y)K_{s-1}(A_r\mid y)p_{s-1}(y) - P_{s-1}(A_h)P_{s-1}(A_r)\right]\\
    \longrightarrow & \sum_{h,r=1}^H c_h c_r \left[\sum_{y\in\mathcal{Y}}K(A_h\mid y)K(A_r\mid y)p(y) - P(A_h)P(A_r)\right]\\
    = & \sum_{h,r=1}^H c_h c_r \Sigma_{hr}
\end{align*}
for $s\rightarrow\infty$.
Moreover, $\sum_{k\ge1}(a_k^2w_k^2)^{-1}\E\{T_k^4\}<\infty$ as $\sum_{k\ge1}(a_k^2 w_k^2)^{-1} = \sum_{k\ge1}w_k^4r_k^2 < \infty$
by assumption and $|T_s|\le|\sum_{h=1}^Hc_h| < \infty$.
Since, by assumption, $b_{n+1}\sum_{k\ge n+1}(a_k b_k^2)^{-1} = r_n\sum_{k>n}w_k^2\longrightarrow1$, then $U_n \rightarrow \sum_{h,r=1}^H c_r c_h \Sigma_{hr}$ $\PP$-a.s.
Therefore, for every $t\in \R$, $\PP$-a.s.,
\begin{equation*}
    \PP\Bigg\{\sqrt{r_n}\sum_{h=1}^H c_h [P(A_h) - P_n(A_h)]\le t \mid Y_{1:n}\Bigg\} = \PP\Big\{M_{n,\infty}\le t \mid Y_{1:n}\Big\} \longrightarrow \Phi\{(\bm c^\top\bm\Sigma\bm c)^{-1/2}t\}.
\end{equation*}

\subsection{Proof of Proposition~\ref{prop:asy_2}}

The proof specializes the proof of \citet[Proposition 2.6]{fortini2024exchangeability} for \textsc{mad} sequences.
By the properties of stable convergence and Proposition~\ref{prop:asy_1}, it is sufficient to show that $\bm\Sigma_n$ is positive definite for $n$ large and converges to $\bm\Sigma$ $\PP$-a.s. with respect to the operator norm.

Since $\bm\Sigma$ is positive definite by assumption, $\bm\Sigma_n$ is positive definite for $n$ large if $\bm\Sigma_n$ converges to $\bm\Sigma$ in the operator norm, which is equivalent to show that
\begin{equation}
    \bm c^\top\bm\Sigma_n\bm c \longrightarrow \bm c^\top\bm\Sigma\bm c
\label{eqn:prop_4_goal}
\end{equation}
$\PP$-a.s. for every $\bm c = (c_1,\ldots,c_H)$ such that $||\bm c|| = 1$.

We can write
\begin{equation*}
    \bm c^\top\bm\Sigma_n\bm c = \sum_{j,r=1}^H c_j c_r \left[\sum_{y\in\mathcal{Y}}K_n(A_j\mid y)K_n(A_r\mid y)p_n(y) - P_n(A_j)P_n(A_r)\right].
\end{equation*}
Then, \eqref{eqn:prop_4_goal} follows since
\begin{equation*}
    \sum_{y\in\mathcal{Y}} K_n(A_j\mid y)K_n(A_r\mid y) p_n(y)
    \longrightarrow \sum_{y\in\mathcal{Y}} K(A_j\mid y)K(A_r\mid y) p(y)
\end{equation*}
by Lemma~\ref{lemma_3} and
\begin{equation*}
    P_n(A_j)P_n(A_r)\longrightarrow P(A_j)P(A_r)
\end{equation*}
by Corollary~\ref{cor:cid}, for every $j,r=1,\ldots,H$.

The final result is obtained by applying the analogue of Slutsky Theorem and Cramer-Wold Theorem for stable convergence (see \citealp[Propositions A2 and A3]{fong2024asymptotics}; and \citealp[Theorem 3.18, Corollaries 3.19 and 6.3]{hausler2015stable}).

\section{Comparison with the discrete copula update}
\label{appendix_copula}

Although our updating using \textsc{mh} kernels shares similarities with the copula update proposed by \cite{fong_holmes_2023}, the two methods also have important differences. Their approach was designed for continuous data, but they proposed a discrete extension
\citep[Section E.1.3]{fong_holmes_2023}.
For $y \in \mathcal{Y}\subseteq\N$, the discrete copula update is
\begin{equation*}
    d_\rho(y,y_n) = 1 - \rho + \rho\frac{\one(y = y_n)}{p_{n-1}(y)},
\end{equation*}
with $\rho\in(0,1)$. Then, for every $n\ge1$ and every $y\in\mathcal{Y}$,
\begin{align*}
    p_{n}(y) & = \{1-w_n+w_n d_\rho(y,y_n)\}p_{n-1}(y)\\
    & = (1-w_n)p_{n-1}(y) + w_n\{(1-\rho)p_{n-1}(y) + \rho\one(y = y_n)\}.
\end{align*}
The induced kernel is a mixture between the old predictive distribution $p_{n-1}(y)$ and a point mass at $y_n$, with the \textsc{dp} update recovered as $\rho\rightarrow1$.
Therefore, the contribution of the new observation to the update is limited to the assigning of mass $w_n\rho$ to $y_n$.
This structure limits the flexibility of the kernel, as it does not allow the borrowing of information between nearby locations. This is clear from Figure~\ref{fig:ill_copula}, which shows the predictive distribution obtained using the discrete copula update for the illustrative example described in Section~\ref{sec:illustration}. 
We employ weights $w_n=(2-n^{-1})(n+1)^{-1}$ both for \textsc{mad} sequence and copula updates.
We denote as \textsc{cop-1} and \textsc{cop-2} the predictive distributions obtained when $\rho$ is selected by maximizing the corresponding prequential log-likelihood and $\rho=0.3$, respectively. The former approach produces a small value of $\hat{\rho}=0.034$, leading to poor learning from the data, too much dependence on the uniform base measure, and small posterior variability. For $\rho=0.3$, the influence of the base measure is reduced, but the predictive distribution lacks smoothing and posterior variability is large.
Conversely, since the \textsc{mh} kernel allows the borrowing of information between nearby values in the support, \textsc{mad} sequences perform better.

\begin{figure}[t]
\centering
\includegraphics[width=0.95\textwidth]{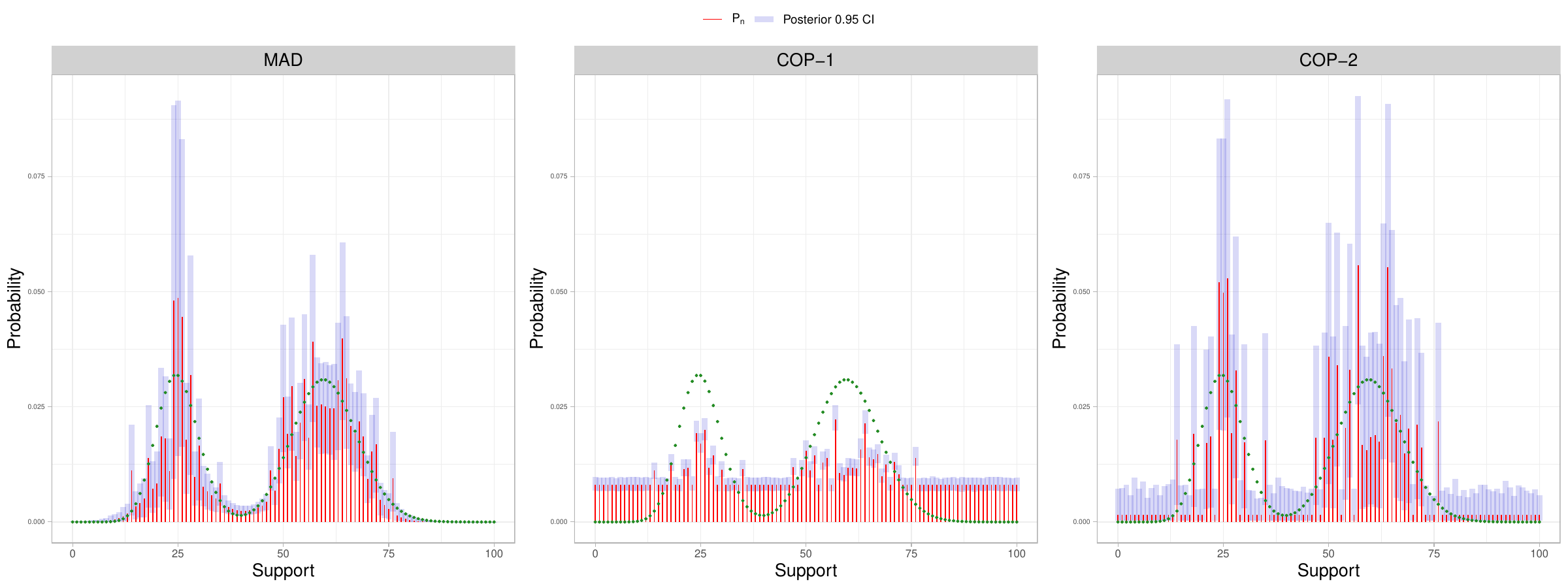}
\caption{\small Predictive distributions and 0.95 credible intervals obtained with \textsc{mad} sequences and two versions of copula updates.}
\label{fig:ill_copula}
\end{figure}

\begin{figure}[t]
\centering
\includegraphics[width=0.9\textwidth]{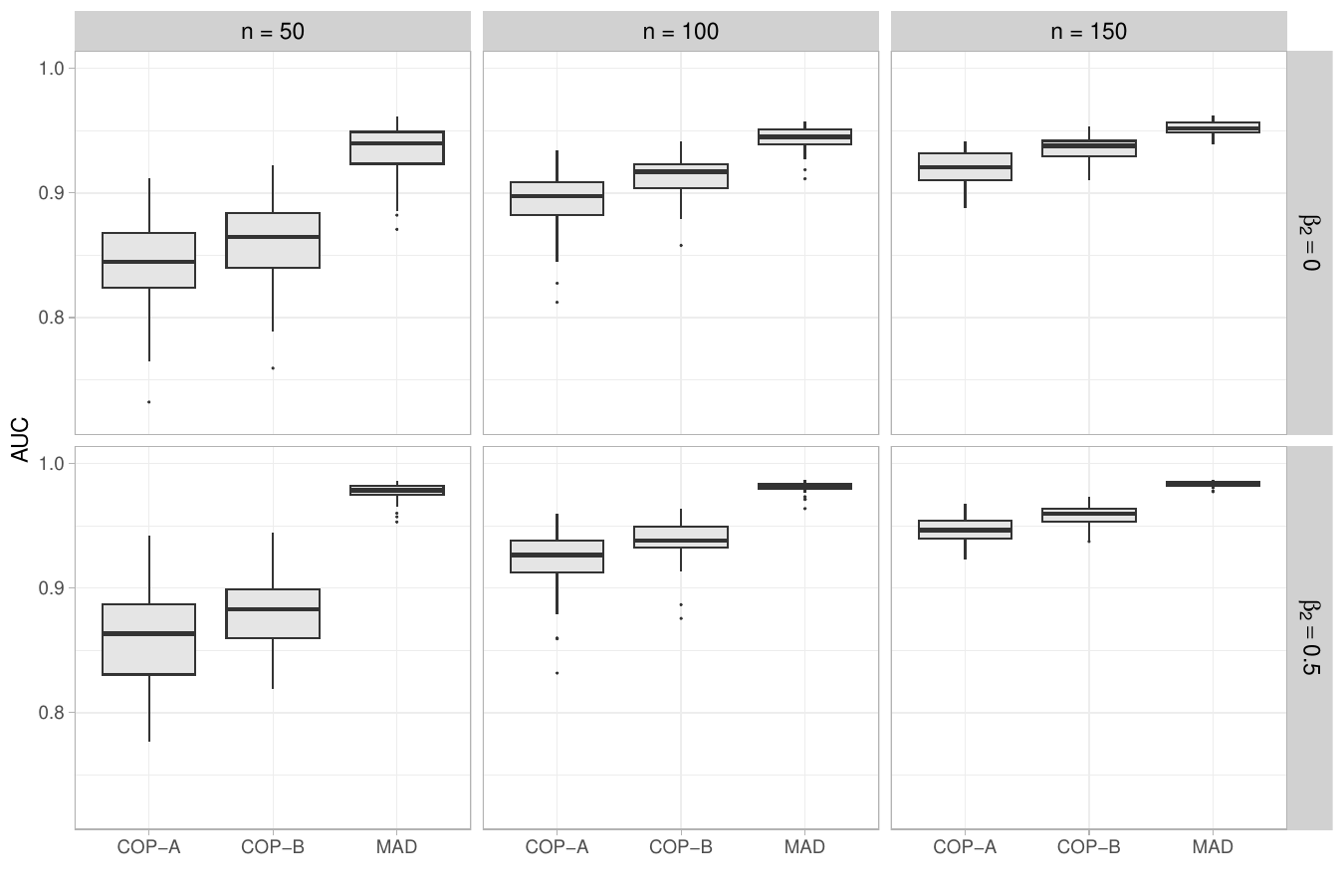}
\caption{\small Results of simulation study comparing \textsc{mad} sequences and two different copula approaches under different data orderings. Box plots show out-of-sample \textsc{auc} obtained on $10^5$ new data points for each scenario.}
\label{fig:ill_copula_2}
\end{figure}

\begin{table}[t]
\centering
\caption{\small Proportion of simulated datasets in which the \textsc{auc} obtained with \textsc{cop-a} is greater than the one obtained with \textsc{cop-b}.}
\label{tab:ill_copula_2}
\vspace{0.2cm}
\footnotesize
\begin{tabular}{l ccc}
& $n=50$ & $n=100$ & $n=150$\\
$\beta_2=0$ & $0.00$ & $0.00$ & $0.00$\\
$\beta_2=0.5$ & $0.06$ & $0.00$ & $0.02$\\
\end{tabular}
\end{table}

We provide a simulation study to compare copula updates with different variable orderings with \textsc{mad} sequences.
We consider a classification problem with two count covariates.
For each sample size $n=\{50,100,150\}$, we generate $50$ simulated datasets.
For each simulation, we sample $X_{i1}$ from the mixture $0.7\mathrm{Poisson}(3) + 0.3\mathrm{Poisson}(12)$, $X_{i2}\sim\mathrm{Poisson}(9)$, and $Y_i\mid x_{i1},x_{i2}\sim\mathrm{Bin}\{1,\mathrm{logit}^{-1}(\eta_i)\}$, where $\eta_i = 6 - 2.1x_{i1} + \beta_2 x_{i2}$, $i=1,\ldots,n$.
We consider settings with $\beta_2=0$ -- i.e. $X^{(2)}$ is a spurious covariate -- and $\beta_2=0.5$, respectively.
We consider two versions of the copula update, \textsc{cop-a} and \textsc{cop-b}, that sequentially update
\begin{equation*}
    X_{i1}, \quad X_{i2}\mid x_{i1}, \quad Y_i\mid x_{i1}, x_{i2},
    \qquad \mathrm{and}\qquad
    X_{i2}, \quad X_{i1}\mid x_{i2}, \quad 
    Y_i\mid x_{i2}, x_{i1},
\end{equation*}
respectively, for $i=1,\ldots,n$.
For all methods, we use weights $w_n=(2-n^{-1})(n+1)^{-1}$ and select the corresponding hyperparameters by maximizing the prequential log-likelihood.
We evaluated the predictive accuracy out of sample using the \textsc{roc} curve computed on $10^5$ new observations and compared the corresponding \textsc{auc}. As reported in Table~\ref{tab:ill_copula_2}, \textsc{cop-b} outperforms \textsc{cop-a} in almost every simulated dataset in the six scenarios. This suggests that predictive performance is sensitive to the ordering of predictor variables. As shown in Figure~\ref{fig:ill_copula_2}, \textsc{mad} sequences consistently have larger \textsc{auc} than both versions of the copula update in every scenario.

\section{Additional Figures}
\label{appendix_figures}

\begin{figure}[H]
\centering
\includegraphics[width=0.9\textwidth]{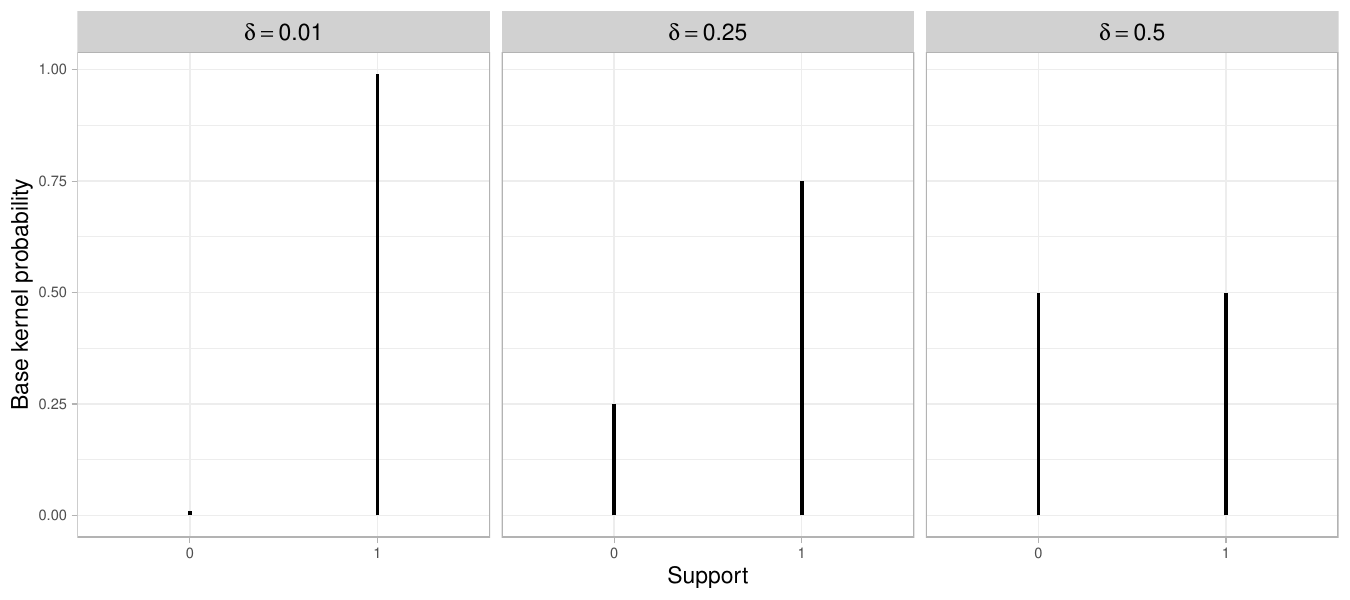}
\caption{\small Base kernel for binary data when $y_n = 1$ and $\delta=\{0.01, 0.25, 0.5\}$.}
\label{fig:binary_base_kernel}
\end{figure}

\begin{figure}[H]
\centering
\includegraphics[width=0.8\textwidth]{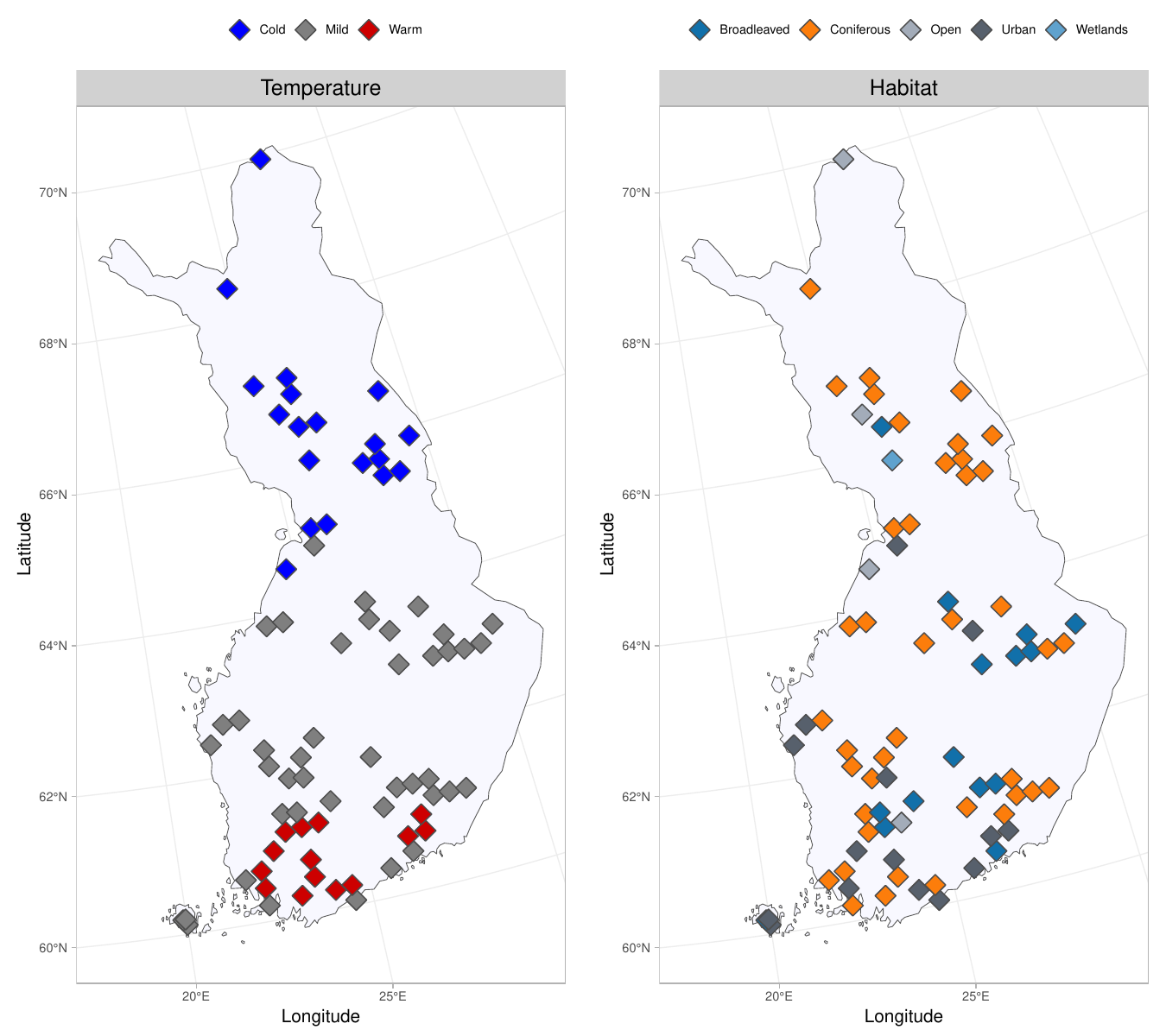}
\caption{\small 
Temperature in April-May (left) and habitat type (right) for each site location.}
\label{fig:corvids_map}
\end{figure}

\begin{figure}[H]
\centering
\includegraphics[width=0.7\textwidth]{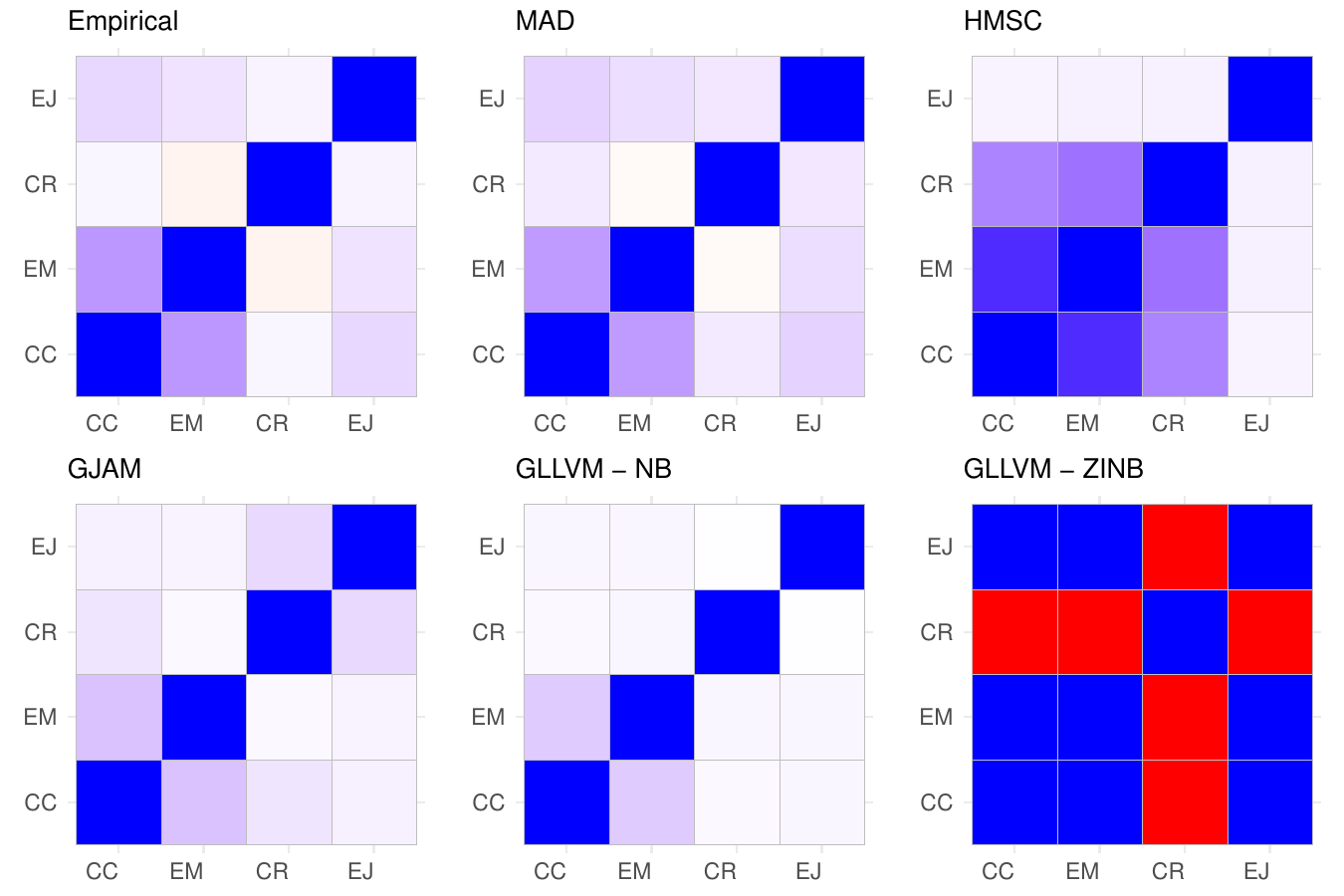}
\caption{\small Inferred correlation between corvid species: carrion crow (\textsc{cc}), Eurasian magpie (\textsc{em}), common raven (\textsc{cr}), and Eurasian jay (\textsc{ej}).  (blue = positive correlation, red = negative correlation, white = zero correlation).}
\label{fig:corvids_3}
\end{figure}

\begin{figure}[H]
\centering
\includegraphics[width=0.65\textwidth]{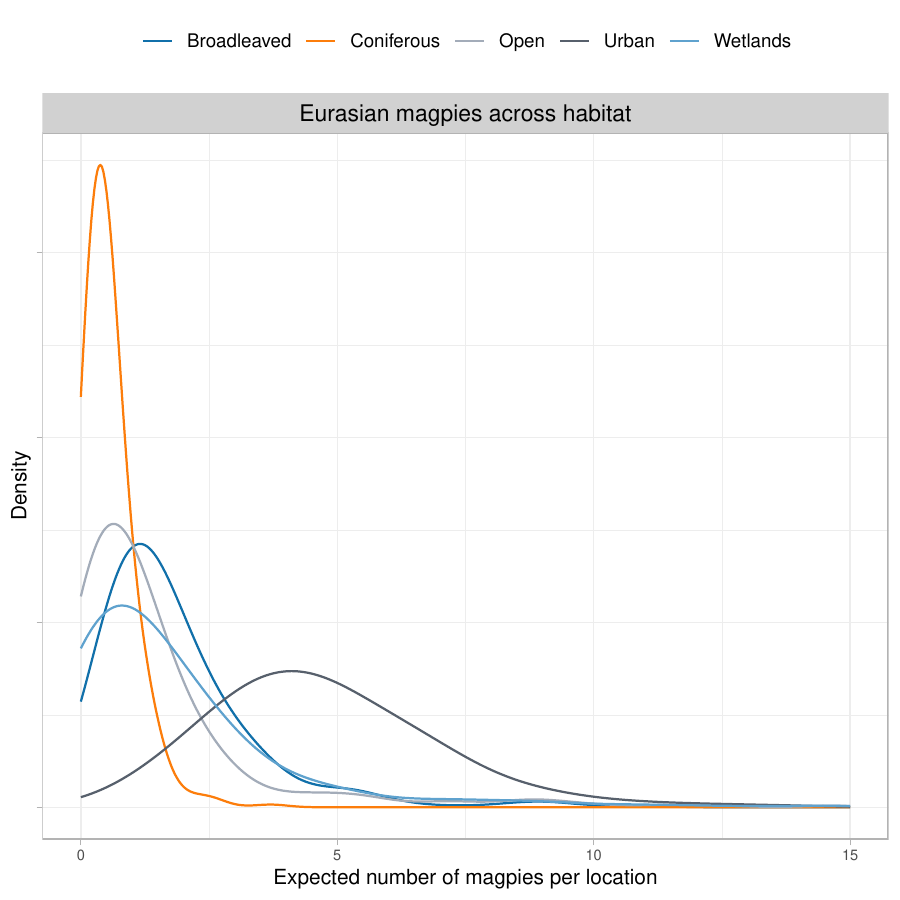}
\caption{\small 
Posterior distribution of the expected number of Eurasian magpies observed at each location in different habitats.}
\label{fig:magpies_habitats}
\end{figure}

\clearpage
\bibliography{biblio}

\end{document}